\documentclass[draft,12pt,reqno,a4paper]{amsart}

\usepackage[margin=3cm,footskip=1cm]{geometry}

\usepackage{amssymb} 
\usepackage{amsmath} 
\usepackage{amsfonts}
\usepackage{amsthm} 
\usepackage{mathtools}
\usepackage{color}

\usepackage{enumerate} 
\usepackage[latin1]{inputenc}
\usepackage[shortlabels]{enumitem}
\usepackage{color}
\usepackage{graphicx} 
\usepackage{verbatim}

\newcommand{\supp}{\operatorname{supp}}

\newcommand{\ad}{{\operatorname{ad}}}
\newcommand{\N}{{\mathbb{N}}}

\newcommand{\R}{{\mathbb{R}}}

\newcommand{\C}{{\mathbb{C}}}
\renewcommand{\S}{{\mathbb{S}}}

\DeclareFontFamily{U}{mathx}{\hyphenchar\font45}
\DeclareFontShape{U}{mathx}{m}{n}{
      <5> <6> <7> <8> <9> <10>
      <10.95> <12> <14.4> <17.28> <20.74> <24.88>
      mathx10
      }{}
\DeclareSymbolFont{mathx}{U}{mathx}{m}{n}
\DeclareFontSubstitution{U}{mathx}{m}{n}
\DeclareMathAccent{\widecheck}{0}{mathx}{"71}

\renewcommand\i{\mathrm{i}}

\newcommand{\p}{{\mathrm p}}

\newcommand{\q}{{\quad}}

\renewcommand{\c}{{\mathrm c}}

\newcommand{\e}{{\mathrm e}}
\newcommand{\ess}{{\mathrm {ess}}}

\renewcommand{\d}{{\mathrm d}}

\newcommand{\pupo}{{\mathrm {pp}}}

\renewcommand{\Re}{\operatorname{Re}}
\renewcommand{\Im}{\operatorname{Im}}

\DeclarePairedDelimiter\inp\langle\rangle

\newcommand\paro[2][]{#1  ( #2#1 )}

\newcommand\parb[2][]{#1 \big ( #2#1\big )}

\newcommand\parbb[2][]{#1 \Big ( #2#1\Big )}
 
 \newcommand{\pp}{{\mathrm {pp}}}
  
\newcommand{\mand}{\text{ \,and\, }}

\DeclarePairedDelimiter\ket{\lvert}{\rangle}
\DeclarePairedDelimiter\bra{\langle}{\rvert}

\DeclareMathOperator*{\slim}{s-lim}

\DeclareMathOperator*{\wvHlim}{{w-\mathcal H}-lim}
\DeclareMathOperator*{\wvLlim}{{w}-{\mathit L^2_{1}}-{lim}}

\DeclareMathOperator*{\vLlim}{{\mathcal L(\mathcal H)}-lim}

\DeclareMathOperator*{\swslim}{s- w^\star-lim}

\DeclarePairedDelimiter\abs\lvert\rvert
\DeclarePairedDelimiter\norm\lVert\rVert
\DeclarePairedDelimiter\set{\{}{\}}

\newcommand{\brR}{{\breve R}}
\newcommand{\brh}{{\breve h}}
\newcommand{\brH}{{\breve H}}
\newcommand{\brT}{{\breve T}}
\newcommand{\brI}{{\breve I}}

\newcommand{\brf}{{\breve f}}
\newcommand{\brr}{{\breve r}}

\newcommand{\brg}{{{\breve \gamma}}}

\newcommand{\brp}{{\breve \psi}}

\newcommand{\bY}{{\mathbf Y}}
\newcommand{\bX}{{\mathbf X}}

\newcommand{\vA}{{\mathcal A}}
\newcommand{\vB}{{\mathcal B}}

\newcommand{\vE}{{\mathcal E}}

\newcommand{\vG}{{\mathcal G}}

\newcommand{\vH}{{\mathcal H}}

\newcommand{\vL}{{\mathcal L}}

\newcommand{\vO}{{\mathcal O}}

\newcommand{\vT}{{\mathcal T}}

\theoremstyle{plain}
\newtheorem{thm}{Theorem}[section]
\newtheorem{defn}[thm]{Definition} 
 
\newtheorem{proposition}[thm]{Proposition}
\newtheorem{lemma}[thm]{Lemma} 
\newtheorem{cond}[thm]{Condition}
\theoremstyle{definition}

\newtheorem{remarks}[thm]{Remarks}
\newtheorem*{remarks*}{Remarks}
\newtheorem*{remark*}{Remark}

\numberwithin{equation}{section}

\title {Green functions and completeness; \\the $3$-body problem
  revisited}

\thanks{
Supported by DFF grant nr.\ 8021-00084B}

\author{E. Skibsted} \address[E. Skibsted]{Institut for Matematik\\
Aarhus Universitet\\ Ny Munkegade 8000 Aarhus C, Denmark}
\email{skibsted@math.au.dk}

\begin{document}

\begin{abstract} Within the class of Derezi{\'n}ski-Enss  pair-potentials
   which includes  Coulomb potentials a stationary scattering theory
   for $N$-body systems was recently developed \cite {Sk1}. In
   particular the wave and
   scattering matrices as well as the restricted wave operators are
   all defined at any    non-threshold energy, and this holds  
   without imposing  any a priori decay condition on  channel eigenstates. In
   this paper we improve for the
   case of $3$-body systems  on the known \emph{weak continuity} properties  in that we show that all non-threshold
   energies are \emph{stationary complete} in this case, resolving a conjecture
   from \cite {Sk1} in the special case  $N=3$. A consequence is that the above
   scattering quantities depend \emph{strongly continuously} on the energy
   parameter  at all non-threshold energies, hence not only almost
   everywhere as previously demonstrated (for an arbitrary $N$). Another consequence is that the
   scattering matrix is unitary at any such  energy. As a side result we give an
   independent     stationary  proof of asymptotic completeness for $3$-body
   systems with long-range  pair-potentials. This is an alternative to the
   known time-dependent proofs \cite{De, En}. 
  \end{abstract}

\allowdisplaybreaks

\maketitle


\tableofcontents

\section{Introduction}\label{sec:Introduction}

In this paper we address a recent conjecture for the stationary
scattering theory of $N$-body systems of quantum particles interacting with
long-range pair-potentials. Thus in the case $N=3$ we can show that
indeed \emph{all}
non-threshold energies are \emph{stationary complete}, resolving  the
problem  posed in  \cite{Sk1} to the affirmative for $N=3$.
The conjecture for particles interacting with short-range
pair-potentials is resolved in  \cite{Sk2}.

Although the bulk of the paper will concern a  more general class
of $3$-body Hamiltonians we will in this introduction confine ourselves to
discussing our results for the standard atomic $3$-body model. We
shall also confine ourselves to  formulations in terms of the atomic Dollard
modification \cite {Do} (see \cite [Remarks 2.2]{Sk1} for a discussion
on how to relate the Dollard
modification to the modification used in  the bulk of the paper). The
paper depends  on several  results of \cite{Sk1} and also on some from  the
more recent works \cite{Sk2,IS4}. Although we do give an account of the
most relevant parts of  \cite{Sk1}, the present paper contains proofs
for which the reader would  benefit from independent parallel
consultance of \cite{Sk1}.

\subsection{Atomic $3$-body model, results}\label{subsec: Atomic 3-body model}
Consider a system of  three  charged particles of dimension $n$
interacting by Coulomb forces. The corresponding 
Hamiltonian  reads 
\begin{equation*}
H=-\sum_{j = 1}^3 \frac{1}{2m_j}\Delta_{x_j} + \sum_{1 \le i<j \le 3} q_iq_j|x_i
-x_j|^{-1}, \quad x_j\in\R^n,\,n\geq 3,
\end{equation*}
where $x_j$,  $m_j$ and $q_j$ denote the position, mass and charge of
the $j$'th particle, respectively.

 The Hamiltonian $H$ is regarded as a self-adjoint operator
on $L^2(\bX)$, where $\bX$ is the $2n$ dimensional real vector
space  $ \{ x=(x_1,x_2,x_3)\mid\sum_{j=1}^{3} m_j x_j = 0\}$.
Let $\vA$  denote the set of all cluster
decompositions of the $3$-particle system. The notation $a_{\max}$ and
$a_{\min}$ refers to the $1$-cluster and $3$-cluster decompositions,
respectively.
  Let for $a\in\vA$ the notation  $\# a$ denote the number of
clusters in $a$.
For $i,j \in\{1, 2, 3\}$, $i< j$, we denote by $(ij) $ the
$2$-cluster decomposition given by letting $C=\{i,j\}$ form a
cluster and the third  particle $l\notin C$ form a singleton. We write $(ij) \leq a$ if $i$ and $j$ belong to the same cluster
in $a$.   More general, we write $b\leq  a$ if each cluster of $b$
is a subset of a cluster of $a$. If $a$ is a $k$-cluster decomposition, $a= (C_1, \dots, C_k)$,
we let
\begin{equation*}
\bX^a = \set[\big]{ x\in\bX\mid \sum_{l\in C_j } m_l x_l = 0,  j = 1, \dots,
k}=\bX^{C_1}\oplus\cdots \oplus\bX^{C_k},
\end{equation*}
and
\[
\bX_a  =  \set[\big]{ x\in\bX\mid  x_i = x_j \mbox{ if } i,j \in C_m  \mbox{ for some }
m \in \{ 1, \dots, k\}  }.
\]
 Note that $b\leq a\Leftrightarrow \bX^b\subseteq\bX^a$, and that the
 subspaces  $\bX^a$ and $\bX_a$  define  an orthogonal decomposition
 of  $\bX$
equipped  with
the quadratic form
$q(x)=\Sigma_j \,2m_j|x_j|^2,  \, x\in {\bX}$.
 Consequently any  $x\in \bX$ decomposes orthogonally as 
 $x =x^{a} + x_{a}$ with $x^a =\pi^a x\in\bX^a$ and $x_a =\pi_a x\in \bX_a$.

With these notations, the $3$-body Schr\"odinger operator
  takes the form $ H = H_0 + V$, 
where  $H_0=p^2$ is (minus)  the Laplace-Beltrami operator on   the
Euclidean space  $(\bX, q)$ and
$V=V(x) =  \sum_{b=(ij)\in\vA} V_{b}(x^{b}) $ with $ V_b (x^b) =
V_{ij} (x_i - x_j)$ for the  $2$-cluster decomposition
$b=(ij)$. Note for example  that 
\begin{equation*}
  x^{(12)}=\parb{\tfrac{m_2}{m_1+m_2}(x_1-x_2),-\tfrac{m_1}{m_1+m_2}(x_1-x_2),0}.
\end{equation*}

More generally for any cluster 
decomposition $a\in\vA$ we introduce a   Hamiltonian $H^a$ as follows. 
For $a=a_{\min}$  we define
$H^{a_{\min}}=0$ on $\mathcal H^{a_{\min}}:=\mathbb C. $
For $a\neq a_{\min}$ we introduce the potential 
\begin{equation*}
V^a(x^a)=\sum_{b=(ij)\leq a} V_{b}(x^b);\quad
x^a\in \bX^a.
\end{equation*} 
Then 
\begin{equation*}
 H^a=-
\Delta_{x^a} +V^a(x^a)=
(p^a)^2 +V^a\ \ 
\text{on }\mathcal H^a=L^2(\bX^a).
\end{equation*}

A channel $\alpha$ is by
definition given as $\alpha =(a,\lambda^\alpha, u^\alpha)$, where
$a\in\vA'=\vA\setminus \{a_{\max}\}$ and  $u^\alpha\in \mathcal H^a$ obeys
$\norm{u^\alpha}=1$ and 
$(H^a-\lambda^\alpha)u^\alpha=0$ for a real number
$\lambda^\alpha$, named a threshold. The set of thresholds is denoted
$\vT(H)$,  and including the eigenvalues of $H$ we introduce
 $\vT_{\p}(H)=\sigma_{\pp}(H)\cup
  \vT(H)$.  For any $a\in \vA'$ the intercluster potential is by definition
\begin{align*}
  I_a(x)=\sum_{b=(ij)\not\leq a}V_b(x^b).
\end{align*} Next we recall  the  atomic Dollard type channel wave operators
\begin{equation}\label{eq:Atomwave_op}
  W_{\alpha,{\rm atom}}^{\pm}=\slim_{t\to \pm\infty}\e^{\i
  tH}\parb{{u^\alpha}\otimes \e^{-\i
  (D_{a,{\rm atom}}^\pm (p_a,t)+\lambda^\alpha t)} (\cdot)},
\end{equation} where
\begin{equation*}
   D_{a,{\rm atom}}^\pm (\xi_a,\pm|t|)=\pm
  D_{a,{\rm atom}}(\pm\xi_a,|t|)\mand D_{a,{\rm atom}}(\xi_a,t)=t\xi_a^2+\int_1^t\,I_a(2s\xi_a)\,\d s.
  \end{equation*} 
 Let $I^\alpha=(\lambda^\alpha,\infty)$ and
$k_\alpha=p_a^2+\lambda^\alpha$. By the intertwining property $H W_{\alpha,{\rm atom}}^{\pm}\supseteq
W_{\alpha,{\rm atom}}^{\pm}k_\alpha $  and the fact that $k_\alpha$ is diagonalized by
the unitary map  $F_\alpha:L^2(\mathbf X_a)\to L^2(I^\alpha
;\vG_a)$, $\vG_a=L^2(\mathbf{S}_a)$,  $\mathbf{S}_a=\mathbf X_a\cap\S^{d_a-1}$ with $d_a=\dim
\mathbf X_a$,  given by
\begin{align*}
  (F_\alpha \varphi)(\lambda,\omega)=(2\pi)^{-d_a/2}2^{-1/2}
  \lambda_\alpha^{(d_a-2)/4}\int \e^{-\i  \lambda^{1/2}_\alpha \omega\cdot
  x_a}\varphi(x_a)\,\d x_a,\quad \lambda_\alpha=\lambda-\lambda^\alpha,
\end{align*} we can for {any} two given  channels $\alpha$ and
$\beta$  write 
\begin{equation*}
  \hat
  S_{\beta\alpha,{\rm atom}}:=F_\beta(W_{\beta,{\rm atom}}^+)^*W_{\alpha,{\rm atom}}^-F_\alpha^{-1}=\int^\oplus_{
  I_{\beta\alpha} }
  S_{\beta\alpha,{\rm atom}}(\lambda)\,\d \lambda,\quad
  I_{\beta\alpha}=I^\beta\cap I^\alpha.
\end{equation*} The fiber operator $ S_{\beta\alpha,{\rm atom}}(\lambda)\in
\vL(\vG_a,\vG_b)$ is from an abstract point of view 
a priori defined only for
a.e. $\lambda \in I_{\beta\alpha}$. It is the
{$\beta\alpha$-entry of the  atomic  Dollard type scattering matrix}
$S_{{\rm atom}}(\lambda)=\parb{S_{\beta\alpha,{\rm
      atom}}(\lambda)}_{\beta\alpha}$ (here the  `dimension' of the
`matrix' $S_{{\rm atom}}(\lambda)$ is $\lambda$-independent on
any interval not containing thresholds).

Introducing the standard notation for weighted spaces 
\begin{align*}
  L_s^2(\mathbf X)=\inp{x}^{-s}L^2(\mathbf X),
\quad s\in\R, \quad\inp{x}=\parb{1+\abs{x}^2}^{1/2},
\end{align*} we recall the following result (here stated for the
atomic $3$-body problem only).

\begin{thm}[\cite {Sk1}]\label{thm:chann-wave-matrD}
  \begin{enumerate}[1)]
  \item 
Let $\alpha$ be a  given channel  $\alpha
    =(a,\lambda^\alpha, u^\alpha)$,     $f:I^\alpha \to \C$ be
  continuous and compactly supported away from  $\vT_{\p}(H)$, and let $s>1/2$. For any $\varphi\in
  L^2(\mathbf  X_a)$   
  \begin{align}\label{eq:wavD1}
  W^\pm_{\alpha,{\rm atom}}
  f\paro{ k_\alpha}\varphi=\int_{I^\alpha \setminus \vT_\p(H)} \,\,f(\lambda)
    W^\pm_{\alpha,{\rm atom}}(\lambda) \paro{F_\alpha \varphi)(\lambda,
  \cdot} 
    \,\d \lambda\in  L_{-s}^2(\bX), 
              \end{align} where the `wave matrices'
              $W^\pm_{\alpha,{\rm
                  atom}}(\lambda)\in\vL\parb{\vG_a,L_{-s}^2(\bX)}$
              with a strongly continuous dependence on $\lambda$.
 In particular for  $\varphi\in
  L_s^2(\mathbf  X_b)$   
the integrand is a continuous compactly supported 
$L_{-s}^2(\bX)$-valued function. In general the integral
has  the weak interpretation of an integral of a measurable 
$L_{-s}^2(\bX)$-valued function.
\item The operator-valued function
  $I_{\beta\alpha}\setminus \vT_{\p}(H)\ni \lambda \to S_{\beta\alpha,{\rm atom}}(\cdot)$ is weakly continuous.  
\end{enumerate}
\end{thm}

Using  the notation $ L^2_\infty(\mathbf X)=\cap_s L^2_s(\mathbf X)$ the
delta-function of $H$ at $\lambda$ is given by 
\begin{equation*}
   \delta(H-\lambda) =\pi^{-1}\Im
{(H-\lambda-\i
0)^{-1}}\text{ as a quadratic form on } L^2_\infty(\mathbf X).
\end{equation*} The adjoint operators $\Gamma^\pm_{\alpha,{\rm atom}}(\lambda)=W^\pm_{\alpha,{\rm atom}}(\lambda)^*$ are
referred to as `restricted channel wave operators'.
\begin{defn}\label{defn:scatEnergy0}  
An  energy $\lambda
  \in \vE:=(\min \vT(H),\infty)\setminus\vT_\p(H)$  is  {stationary
    complete} for $H$ if  
\begin{equation}\label{eq:ScatEnergy22233}
  \forall \psi\in  L^2_\infty(\mathbf X):\,\, \sum_{\lambda^\beta<
  \lambda}\,\norm{\Gamma^\pm_{\beta,{\rm atom}}(\lambda) \psi}^2= 
  \inp{\psi,\delta(H-\lambda){\psi}}.
\end{equation} 
  \end{defn}

The main result of the  present paper for the atomic $3$-body
Hamiltonian (obtained independently of the asymptotic completeness
property known from  
 \cite{De,En}) reads as follows.
\begin{thm}\label{thm:Parsconcl-gener0} All $\lambda
  \in \vE$ are stationary complete for the  $3$-body
Hamiltonian $H$.
  \end{thm} Although we shall not elaborate  we remark  that it  is
  here possibly essentially  to replace
$\vT_\p(H)$  by $\vT(H)$ (see a  discussion in Subsection \ref{subsec:Conclusion and generalizations}).

Note that while
\eqref{eq:wavD1} may be taken as a definition (although implicit) of the wave matrices,
\eqref{eq:ScatEnergy22233} is a non-trivial derived property.  It is known
from  \cite{Sk1} that Lebesgue almost all non-threshold energies  are  stationary
complete for the  atomic $N$-body
Hamiltonian. However it is
an  open problem to show stationary completeness at fixed energy  for
the atomic $N$-body  model in
the case  
$N\geq 4$.

One can regard \eqref{eq:ScatEnergy22233} as an `on-shell  Parseval
formula'. By integration it implies asymptotic completeness, hence
providing an alternative to the proofs of  \cite{De, En}. Note that
the existence of the channel wave operators and 
Theorem \ref{thm:chann-wave-matrD} can be shown independently 
of time-dependent methods (see  Remarks \ref{remark:R}). Moreover 
there are  immediate consequences  for the discussed scattering
quantities considered as operator-valued functions on $\vE$ (recalled
in Subsection \ref{subsec: -body effective potential and a $1$-body radial limit}):
\begin{enumerate}[I)]
\item \label{item:stroS} The scattering matrix $S_{\rm
      atom}(\cdot)$ is a  strongly continuous unitary operator
    determined uniquely  by asymptotics of minimum
  generalized eigenfunctions (at any given  energy). The   latter are
  taken from  the ranges
  of the wave matrices (at this energy). 
\item The  restricted channel  wave operators $\Gamma^\pm_{\alpha,{\rm atom}}(\lambda)$ are strongly continuous.
\item  The scattering matrix links the incoming and outgoing wave 
  matrices,
\begin{equation*}
     W^-_{\alpha,{\rm
      atom}}(\lambda)=\sum_{\lambda^\beta<\lambda} W^+_{\beta,{\rm
      atom}}(\lambda)S_{\beta\alpha,{\rm
      atom}}(\lambda);\quad \lambda \in \vE,\, \lambda^\alpha <\lambda.
  \end{equation*}
\end{enumerate}

Our proof of \eqref{eq:ScatEnergy22233} relies on a characterization
of this property from  \cite {Sk1} (recalled in
\eqref{eq:asres29}).  In particular we derive the top-order asymptotics of any vector
on the form $(H-\lambda-\i
0)^{-1}\psi$,  $\psi\in  L^2_\infty(\mathbf X)$, yielding
\eqref{eq:ScatEnergy22233}. 
 Such an  asymptotics  is also
appearing  in the $N$-body setting of \cite{Sk1}, however there only
proven away from a Lebesgue null-set of energies.

\subsection{Extensions, comparison with the literature and  discussion}\label{subsec: Extensions and comparison with the literature}
In the above exposition we have used  Coulombic pair-potentials
motivated by physics. However the results are proven for a more
general class of pair-potentials $V_a=V_a(x^a)$. We need the decay
$V_a=\vO(\abs{x^a}^{-\mu})$  with $\mu>\sqrt 3 -1$   and similar decay
conditions for higher derivatives (assuming here smoothness outside a
compact set), which includes the
Coulomb potential. The critical exponent  $\sqrt 3 -1$  agrees with the
one of \cite{De, En}.

There is a fairly big literature on  stationary scattering theory
for $3$-body systems both on the mathematical side  and the physics side. This
is to a large extent  based on the Faddeev method or some modification of
that, see for example \cite{GM,Ne,Me}. The Faddeev method, as for
example used in the mathematically rigorous paper \cite{GM}, requires
fall-off like $V_a=\vO(\abs{x^a}^{-2-\epsilon})$. Moreover  there are
additional complications in that the threshold zero needs be be
regular for the two-body systems (i.e.  zero-energy eigenvalues
and  resonances are excluded) and  `spurious poles' cannot be ruled 
out  (these poles
would arise from lack of solvability of a certain Lippmann--Schwinger
type equation). On the physics side the $3$-body problem with Coulombic pair-potentials
 has attracted much attention, see for example \cite{Me} and the works cited
there. The picture seems to be the same,  modified Faddeev type
methods involve implicit conditions at zero energy for the two-body
systems and spurious poles cannot be excluded.

The work on the $3$-body stationary scattering theory  
\cite{Is3} (see also its partial precursor
\cite{Is2} or the recent book \cite{Is5})  is different. In fact Isozaki does not assume any
regularity at zero energy for the two-body
systems  and his theory does not have spurious
poles. He overcomes these deficencies by avoiding the otherwise
prevailing Faddeev method. On the other hand \cite{Is3} still
needs, in some comparison argument, a very detailed  information on the
spectral theory of the (two-body) sub-Hamiltonians at zero energy, and
this requires the  fall-off condition 
$V_a=\vO(\abs{x^a}^{-5-\epsilon})$ (as well as a restriction on the
particle dimension). The  present paper as well as our previous paper
for the short-range case \cite{Sk2} do not
involve such detailed  information. In fact in  comparison with
Isozaki's papers    Borel
calculus arguments  suffice. 

Otherwise the overall spirit of the
present paper, \cite{Sk2} 
and \cite{Is3} is  the same, in particular the use of resolvent
equations are  kept at a minimum (solvability of Lippmann--Schwinger
type equations is not an issue) and these  works employ intensively
 Besov spaces not only in proofs but also in the  formulation of
 various results.
In conclusion we  recover all of Isozaki's results 
by a different method that works down  to the critical exponent $\sqrt
3 -1$. In particular our theory  covers pair-potentials with  Coulombic
asymptotics without using any implicit condition, as elaborated on in 
 Subsection \ref{subsec: Atomic 3-body model}. In comparison with 
 \cite{Sk2} that got the improvements  down to the critical exponent $1$
 (in fact for all $N\geq 3$),  the long-range case is considerably
 more complicated. It is still an open problem to show
 stationary  completeness at any $ \lambda
  \in \vE$ for $N\geq 4$ in this case.

We remark that the strong continuity assertion for the scattering
matrix, cf. 
\ref{item:stroS}, cannot in general be replaced by norm continuity, see
\cite[Subsection 7.6]{Ya1} for a counterexample  with a short-range
potential. In this sense the stated regularity of the scattering
matrix is optimal. 

We also remark that from a physics point of view well-definedness and continuity of basic
scattering quantities is only one out of several relevant problems, for
example like the structure of the singularities of the kernel of the scattering
matrix (possibly considered as function of the  energy parameter), see
for example \cite{Me, Is1, Is5, SW} for results on 
this topic. This and related topics  go beyond the scope of the present
paper, although  our formulas potentially could reveal some insight.

\section{Preliminaries}\label{sec:preliminaries}

We explain our setting  and give an account of a number of results from
\cite{Sk1} on the general $N$-body stationary theory. 
\subsection{$N$-body Hamiltonians, assumptions  and
  notation}\label{subsec:body Hamiltonians, limiting absorption
  principle  and notation}

Let $\bX$ be                    
a (nonzero) finite dimensional real inner product space,
equipped with a
finite family $\{\bX_a\}_{a\in \vA}$ of subspaces closed under intersection:
For any $a,b\in\mathcal A$ there exists $c\in\mathcal A$ such that $\bX_a\cap\bX_b=\bX_c.$
  We
 order $\vA$ by  writing $a\leq b$ (or equivalently as $b\geq a$) if
$\bX_a\supseteq \bX_b$. 
It is assumed that there exist
$a_{\min},a_{\max}\in \vA$ such that 
$\bX_{a_{\min}}=\bX$ and 
$\bX_{a_{\max}}=\{0\}$. The subspaces $\bX_a$, $a\neq a_{\min} $,  are called  \emph{collision
 planes}. We will  use the notation $d_a=\dim
\mathbf X_a$.  In  Section \ref{sec:Stationary
  completeness for the 3} it will be convenient to  use
  the abbreviated  notations  $d=d_{a_{\min}}=\dim \mathbf X$ and $a_0=a_{\max}$.
 
The $2$-body model (or more correctly named `the one-body model')  is
based on the structure 
 $\vA=\set{a_{\min},a_{\max}}$. 
The scattering  theory for such models is well understood,
in fact there are several doable approaches under the below Condition
\ref{cond:smooth2wea3n12} \eqref{item:shortr} and \eqref{item:shortl}, see for example \cite[Chapter
4]{DG} and
\cite{II, IS4}  for  accounts  on time-dependent  and stationary
 long-range  scattering  theories, respectively. 

Introducing
 \begin{equation*}
  \vA_1=\vA\setminus
  \{a_{\max}\} 
\mand \vA_2=\vA\setminus
  \{a_{\min},a_{\max}\},
 \end{equation*} 
the $3$-body model is based on the structure 
\begin{equation}\label{eq:3body}
  \vA_2 \neq \emptyset \mand \bX_a\cap\bX_b=\set{0};\quad a,b\in \vA_2\mand  a\neq
  b.
\end{equation} This condition will be imposed in Section \ref{sec:Stationary completeness for the 3}.

Let $\bX^a\subseteq\bX$ be the orthogonal complement of $\bX_a\subseteq \bX$,
and denote the associated orthogonal decomposition of $x\in\bX$ by 
$$x=x^a\oplus x_a=\pi^ax\oplus \pi_ax\in\bX^a\oplus \bX_a.$$ 
The vectors $x_a$ and $x^a$ may be  called the \emph{inter-cluster} and 
\emph{internal components} of $x$, respectively. 

A real-valued measurable function $V\colon\bX\to\mathbb R$ is 
a \textit{potential of many-body type} 
if there exist real-valued measurable functions
$V_a\colon\bX^a\to\mathbb R$ such that 
\begin{equation*}
V(x)=\sum_{a\in\mathcal A}V_a(x^a)\ \ \text{for }x\in\mathbf X.
\end{equation*} We take $V_{a_{\min}}=0$ (without loss of
generality). We impose throughout the paper the
following condition from  \cite{Sk1}. By definition $\N_0=\N\cup\set{0}$.
\begin{cond}\label{cond:smooth2wea3n12}
    There exists $\mu\in (0,1)$  such that for all $a\in \vA\setminus\set{a_{\min}}$ the
    potential $V_a(x^a)=V^a_{\rm sr}(x^a)+V^a_{\rm lr}(x^a)$, where
    \begin{enumerate}
\item \label{item:shortr}$V^a_{\rm sr}(-\Delta_{x^a}+1)^{-1}$ is compact and  $\abs{x^a}^{1+\mu}
      V^a_{\rm sr}(-\Delta_{x^a}+1)^{-1}$ is bounded.
    \item \label{item:shortl}$V^a_{\rm lr}\in C^\infty$ and for all $\gamma\in \N_0^{\dim \mathbf X^a}$
      \begin{equation*}
        \partial^\gamma V^a_{\rm lr}(x^a)=\vO(\abs{x^a}^{-\mu-|\gamma|}).
      \end{equation*}
\item \label{item:DE} $\mu>\sqrt 3-1$.
\end{enumerate}
\end{cond}
\begin{remarks*}
  \begin{enumerate}[i)]
  \item \label{item:Di1}
   The third condition \eqref{item:DE} coincides with a requirement
  of the long-range theories
  \cite{En, De,Sk1}. All theories   work well with  \eqref{item:DE}  in combination with
  \eqref{item:shortr} and \eqref{item:shortl}. Although there are
  long-range theories for $N$-body models with $N\geq 3$ for which
  the condition \eqref{item:DE} is not fulfilled (see for example
  \cite{Is5} for
  references), they all require (as far as  the author
  knows) either decay assumptions on  sub-Hamiltonian 
  eigenstates or geometric assumptions on the pair-potentials (including a 
  sign condition at infinity).
\item \label{item:Di2} This paper depends on \cite{Sk1}, however let
  us remark that all relevant results from \cite{Sk1} clearly extend
   to the setting, where $V^a_{\rm lr}\in C^\infty$ in
   \eqref{item:shortl} is replaced by $V^a_{\rm lr}\in C^l$ for a big
   enough $l\in \N$. In the terminology of \cite{IS4} such potential
   is a \emph{classical $C^l$ long-range potential} and the results of
   \cite{IS4} are at disposal (some of them  will  be crucial for us in Subsection \ref{subsec:Free channel
  contribution}). Part of \cite{Sk1} depends on a stationary phase
argument from \cite{II}, which requires this $l$
to be sufficiently large. On the other hand the one-body setup of
\cite{IS4} (involving position-space wave operators rather than momentum-space wave operators) requires
only $l=2$, so it natural to expect modified versions of \cite{Sk1}  and the
present paper requiring only  classical $C^2$ long-range pair-potentials, in
fact (given  \eqref{item:DE})  such modified theories for $C^1$ long-range
pair-potentials upon using
Dollard-type channel wave operators. Such modifications  will not be
presented or examined in the present paper.
\end{enumerate}
\end{remarks*}

 For any $a\in\vA$  we  introduce an  associated  {Hamiltonian} $H^a$   as follows. 
For $a=a_{\min}$  we define
$H^{a_{\min}}=0$ on $\mathcal H^{a_{\min}}=L^2(\set{0})=\mathbb C. $
For $a\neq a_{\min}$ 
we let 
\begin{equation*}
V^a(x^a)=\sum_{b\leq a} V_{b}(x^b)
,\quad
x^a\in \bX^a,
\end{equation*} 
and  introduce then 
\begin{align*}
 H^a=-
\Delta_{x^a} +V^a\ \ 
\text{on }\mathcal H^a=L^2(\bX^a).
\end{align*} 
We abbreviate 
\begin{align*}
V^{a_{\max}}=V,\quad
 H^{a_{\max}}=H,\quad 
 \mathcal H^{a_{\max}}=\mathcal H.
\end{align*}
 The operator $H$ (with domain $
\mathcal D(H)=H^2(\mathbf X)$) is  the full
Hamiltonian of the $N$-body model,  and the \textit{thresholds} of $H$ are by
definition the
eigenvalues of the  sub-Hamiltonians $H^a$; 
$a\in  \vA_1$.
 Equivalently stated  the set of thresholds is 
\begin{equation*}
 \vT (H):= \bigcup_{a\in\vA_1} \sigma_{\pupo}( H^a).
\end{equation*}
 This set 
is closed and countable. Moreover the set of non-threshold eigenvalues
  is discrete in $\R\setminus \vT (H)$,  and it 
 can only  accumulate  at  points in
$\vT (H)$  from below.  
The essential spectrum is given by the formula
$\sigma_{\ess}(H)= \bigl[\min \vT(H),\infty\bigr)$. 
We introduce the notation $\vT_{\p}(H)=\sigma_{\pp}(H)\cup
  \vT(H)$, and more generally   $\vT_{\p}(H^a)=\sigma_{\pp}(H^a)\cup
  \vT(H^a)$. Denote $R(z)=(H-z)^{-1}$ for
$z\notin \sigma(H)$.

Consider   and fix $\chi\in C^\infty(\mathbb{R})$ such that 
\begin{align*}
\chi(t)
=\left\{\begin{array}{ll}
0 &\mbox{ for } t \le 4/3, \\
1 &\mbox{ for } t \ge 5/3,
\end{array}
\right.
\quad
\chi'\geq  0,
\end{align*} and such that the following properties are fulfilled:
\begin{align*}
  \sqrt{\chi}, \sqrt{\chi'}, (1-\chi^2)^{1/4} ,
  \sqrt{-\parb{(1-\chi^2)^{1/2} }'}\in C^\infty.
\end{align*} We define correspondingly $\chi_+=\chi$ and
$\chi_-=(1-\chi^2)^{1/2} $ and record that
\begin{align*}
  \chi_+^2+\chi_-^2=1\mand\sqrt{\chi_+}, \sqrt{\chi_+'}, \sqrt{\chi_-}, \sqrt{-\chi_-'}\in C^\infty.
\end{align*}

Any  function $f\in C^\infty_\c(\R)$ taking values in $[0,1]$ is
referred to as a  \emph{standard support function} (or just a `support function'). For  any such
functions  $f_1$ and $f_2$ we write
$f_2\succ f_1$, if $f_2=1$ in a neighbourhood of $\supp f_1$.

We shall use the notation 
$\inp{x}:=\parb{1+\abs{x}^2}^{1/2}$ for  $x\in \mathbf X$ (or more generally for
any $x$ in a normed space).  If $T$ is a self-adjoint  operator on a
Hilbert space $\vG$ and $\varphi\in \vG$ then
$\inp{T}_\varphi:=\inp{\varphi,T\varphi}$. We denote the space of
bounded (linear) operators 
from one  (general) Banach space $X$ to another one $Y$ by $\vL(X,Y)$ 
and abbreviate $\mathcal L(X)=\mathcal L(X,X)$. The dual space of $X$
is denoted by $X^*$.

To define \emph{Besov spaces 
associated with the multiplication operator
$|x|$ on $\vH$}  
let
\begin{align*}
F_0&=F\bigl(\bigl\{ x\in \mathbf X\,\big|\,\abs{x}<1\bigr\}\bigr),\\
F_m&=F\bigl(\bigl\{ x\in \mathbf X\,\big|\,2^{m-1}\le \abs{x}<2^m\bigr\} \bigr)
\quad \text{for }m=1,2,\dots,
\end{align*}
where $F(U)=F_U$ is the sharp characteristic function of any given  subset
$U\subseteq {\mathbf X}$. 
The Besov spaces $\mathcal B =\mathcal B(\mathbf X)$, $\mathcal
B^*=\mathcal B(\mathbf X)^*$ and $\mathcal B^*_0=\mathcal
B^*_0(\mathbf X)$ are then given  as 
\begin{align*}
\mathcal B&=
\bigl\{\psi\in L^2_{\mathrm{loc}}(\mathbf X)\,\big|\,\|\psi\|_{\mathcal B}<\infty\bigr\},\quad 
\|\psi\|_{\mathcal B}=\sum_{m=0}^\infty 2^{m/2}
\|F_m\psi\|_{{\mathcal H}},\\
\mathcal B^*&=
\bigl\{\psi\in L^2_{\mathrm{loc}}(\mathbf X)\,\big|\, \|\psi\|_{\mathcal B^*}<\infty\bigr\},\quad 
\|\psi\|_{\mathcal B^*}=\sup_{m\ge 0}2^{-m/2}\|F_m\psi\|_{{\mathcal H}},
\\
\mathcal B^*_0
&=
\Bigl\{\psi\in \mathcal B^*\,\Big|\, \lim_{m\to\infty}2^{-m/2}\|F_m\psi\|_{{\mathcal H}}=0\Bigr\},
\end{align*}
respectively.
Denote the standard \emph{weighted $L^2$ spaces} by 
$$
L_s^2=L_s^2(\mathbf X)=\inp{x}^{-s}L^2(\mathbf X)\ \ \text{for }s\in\mathbb R ,\quad
L_{-\infty}^2=\bigcup_{s\in\R}L^2_s,\quad
L^2_\infty=\bigcap_{s\in\mathbb R}L_s^2.
$$ 
Then for any $s>1/2$
\begin{equation*}
 L^2_s\subsetneq \mathcal B\subsetneq L^2_{1/2}
\subsetneq \mathcal H
\subsetneq L^2_{-1/2}\subsetneq \mathcal B^*_0\subsetneq \mathcal B^*\subsetneq L^2_{-s}.
\end{equation*} The abstract quotient-norm on the Banach space
$\vB^*/\vB_0^*$ and 
  \begin{equation*}
    \norm{\psi}_{\rm quo}:=\limsup_{n\to \infty}\,2^{-n/2}\norm[\Big]{\sum_{m=0}^n F_m\psi}_{{\mathcal H}},\quad \psi\in \vB^*,
  \end{equation*} are 
equivalent   norms.

We recall the following notion of
order of decay \cite[(6.2)]{Sk1}: 
An operator $T$   on $\mathcal H$ such that
$T,T^*:L^2_\infty\to L^2_\infty$ is of  \emph{order}
$t\in\mathbb R$, written $T=\vO(\inp{x}^t)$, if 
 for each $s\in\mathbb R$  the restriction  $T_{|L^2_\infty}$ extends to
 an operator $T_s\in\vL(L^2_{s}, L^2_{s-t})$. If  $T$ has {order} $t$
 for all 
$t\in\mathbb R$, we write  $T=\vO(\inp{x}^{-\infty})$.

\begin{subequations}
Under a rather weak condition (in particular weaker  than Condition
\ref{cond:smooth2wea3n12}) it is demonstrated in  \cite{AIIS} that the
following limits exist   locally
                                                              uniformly
                                                              in $\lambda\not\in \vT_{\p}(H)$:
                                                              \begin{align}\label{eq:LAPbnda}
  R(\lambda\pm \i
  0)=\lim _{\epsilon\to 0_+} \,R(\lambda\pm\i
    \epsilon)\in \vL\parb{L^2_s,L^2_{-s}}\text{ for any }s>1/2.
\end{align} Furthermore, 
   in the strong weak$^*$-topology, 
  \begin{align}\label{eq:BB^*a}
    \begin{split}
      R(\lambda\pm \i
  0)&=\swslim _{\epsilon\to 0_+} \,R(\lambda\pm\i
    \epsilon)\in \vL\parb{\vB,\vB^*}\\& \text{
                                                              with a
                                                              locally
                                                              uniform
                                                              norm bound in
                                                              }\lambda\not
                                                              \in \vT_{\p}(H).
    \end{split}
\end{align}
\end{subequations}

\subsection{One-body effective potentials and  $N$-body scattering theory}\label{subsec: -body effective potential and a $1$-body radial limit}

For any  $a\in  \vA_1$ we introduce   $I^{\rm
  sr}_a=\sum_{b\not\leq a}V_{\rm sr}^b$, $I^{\rm lr}_a=\sum_{b\not\leq
  a}V_{\rm lr}^b$, $I_a=I^{\rm sr}_a+I^{\rm
  lr}_a$  and
\begin{subequations}
 \begin{equation}\label{eq:Ipott}
  \breve I_{a,R}=\breve I_{a,R}(x_a)=\chi_+(|x_a|/R)I^{\rm lr}_a(x_a)\prod_{b\not\leq a} \,\,\chi_+(|\pi^b
  x_a| \ln \inp{x_a}/\inp{x_a} );\quad R\geq 1.
\end{equation}  
We note that if  the $3$-body condition \eqref{eq:3body} is imposed
the  last (product)  factor can be taken to be one except for
$a=a_{\min}$. For $a=a_{\min}$ and for the general  $N$-body problem considered in this
section  the factor is needed to provide fall-off.
 More precisely  the `regularization' $\breve I_{a,R}$  is  a one-body potential fulfilling for any
$\breve\mu\in (0,\mu)$ the bounds 
\begin{equation}\label{eq:brevePotentialOne}
       \partial^\gamma \breve I_{a,R}(x_a)=\vO(\abs{x_a}^{-\breve\mu-|\gamma|}).
      \end{equation}  
\end{subequations}
  For notational
     convenience we take from this point and throughout the paper $\breve\mu=\mu$,
     i.e. more precisely we will assume \eqref{eq:brevePotentialOne} with
     $\breve\mu$ replaced by $\mu$.  For all of our main  results  $\breve I_{a,R}$ enters
     only with 
     $R=1$, and in that case we abbreviate  $ \breve I_{a}=\breve
     I_{a,1}$. We explain   usages of  the auxiliary one-body potential $\breve I_{a,R}$ for   $R\geq 1$
     taken large 
     in Remarks \ref{remark:R}.

We let 
      $\breve K_a(\cdot,\lambda)$,  $\lambda>0$, denote the corresponding approximate solution
      to the eikonal equation $\abs{\nabla_{x_a} \breve K_a}^2+ \breve
      I_a=\lambda$ as  taken 
       from \cite {Is,II}. 
More precisely writing  $\breve K_a(x_a,\lambda)=\sqrt
      \lambda \abs{x_a}-\breve k_a(x_a,\lambda)$  the following properties
      are fulfilled with $\R_+:=(0,\infty)$ and  $ d_a:=\dim \mathbf
      X_a$. The functions $\breve k_a\in
      C^\infty(\mathbf X_a \times \R_+)$ and:
      \begin{enumerate}[1)]
      \item For any compact $\Lambda \subset \R_+$ there exists $\rho>1$
        such that for all $\lambda\in \Lambda$ and all $x_a\in \mathbf X_a$
        with $\abs{x_a}>\rho$
        \begin{align*}
          2\sqrt{\lambda}\,\tfrac{\partial} {\partial\abs{x_a}}\breve
          k_a=\breve I_a(x_a)+\abs{\nabla_{x_a}\breve k_a}^2.
        \end{align*}
      \item For all multiindices $\gamma\in \N_0^{d_a}$, $m\in \N_0$ and
         compact  $\Lambda \subset \R_+$
        \begin{align*}
          \abs[\big]{{\partial}_
          {{x_a}}^\gamma {\partial}_
          {\lambda}^m \breve k_a}\leq C\inp {x_a}^{1-\abs{\gamma}-\mu}\text{
          uniformly in }\lambda\in \Lambda.
        \end{align*}
      \end{enumerate}

Next we apply the  Legendre transform of $\breve K_a$  following \cite[Lemma 6.1]{II}: 
 There exist an $\mathbf X_a$-valued function $x(\xi, t)$ and a
positive function $\lambda(\xi,t)$ both in $C^\infty\parb{(\mathbf
X_a\setminus\set{0})\times \R_+}$ and satisfying the following requirements.
 For any compact set $B\subseteq \mathbf
X_a\setminus\set{0}$,  there exist $ T, C > 0$ such  that for $\xi\in B$
and $t > T$ 
\begin{enumerate}[1)]
      \item\label{item:1ss} \quad $\xi=\partial_{x_a}\breve
        K_a\parb{x(\xi, t),\lambda(\xi,t)},\quad t=\partial_{\lambda}
        \breve K_a\parb{x(\xi, t),\lambda(\xi,t)},$
        
      \item \label{item:2ss}\quad $\abs{x(\xi, t)-2t\xi}\leq C\inp{t}^{1-\mu},\quad \abs[\big]{\lambda(\xi, t)-\abs{\xi}^2}\leq C\inp{t}^{-\mu}.$
\end{enumerate} Then we define 
\begin{align*}
  \breve S_a(\xi,t)=x(\xi, t)\cdot \xi +\lambda(\xi, t)t - \breve K_a(x(\xi, t),
  \lambda(\xi, t));\quad (\xi,t)\in\parb{\mathbf
X_a\setminus{0}}\times \R_+.
\end{align*} 
 Note that this function solves the Hamilton-Jacobi equation
\begin{align*}
  \partial_t \breve  S_a(\xi,t)=
  \xi^2+\brI_a\parb{\partial_{\xi} \breve S_a(\xi,t)};\quad t> t(\xi), \, \xi\neq 0.
\end{align*}

Consider now any  \emph{channel}
$\alpha =(a,\lambda^\alpha, u^\alpha)$, i.e.   $a\in\vA_1$,  $u^\alpha\in
\mathcal H^a$ and 
$(H^a-\lambda^\alpha)u^\alpha=0$. We introduce the corresponding \emph{channel wave operators} 
\begin{equation}\label{eq:wave_op}
  W_\alpha^{\pm}=\slim_{t\to \pm\infty}\e^{\i tH}J_\alpha\e^{-\i
  (\breve  S_a^\pm (p_a,t)+\lambda^\alpha t)},\quad J_\alpha \varphi
=u^\alpha\otimes
\varphi,
 \end{equation} where $p_a=-\i \nabla_{x_a}$ and $\breve S_a^\pm
 (\xi_a,\pm|t|)=\pm \breve S_a(\pm\xi_a,|t|)$.   The existence of these
 limits can be proven independently of \cite{De,En}, see  Remarks
 \ref{remark:R}. It is a general fact  that the existence of the
wave operators  implies their orthogonality, see for example
\cite[Theorem XI.36]{RS}.

Let us for  the channel
$\alpha =(a,\lambda^\alpha, u^\alpha)$ introduce the notation
\begin{equation*}
k_\alpha=p_a^2+\lambda^\alpha, \quad I^\alpha=(\lambda^\alpha,\infty)\quad\text{and}\quad \vE^\alpha=I^\alpha \setminus {
  \vT_{\p}(H)}.
\end{equation*} 
 Note  the intertwining property $H W_\alpha^{\pm}\supseteq
W_\alpha^{\pm}k_\alpha $  and the fact that $k_\alpha$ is diagonalized by
the unitary map  $F_\alpha:L^2(\mathbf X_a)\to L^2(I^\alpha
;\vG_a)$, $\vG_a=L^2(\mathbf{S}_a)$,  $\mathbf{S}_a=\mathbf X_a\cap\S^{d_a-1}$ with $d_a=\dim
\mathbf X_a$,  given by
\begin{align}\label{eq:Four}
  \begin{split}
  (F_\alpha \varphi)(\lambda,\omega)&=(2\pi)^{-d_a/2}2^{-1/2}
  \lambda_\alpha^{(d_a-2)/4}\int \e^{-\i  \lambda^{1/2}_\alpha \omega\cdot
  x_a}\varphi(x_a)\,\d
x_a;\\&\quad \quad\lambda_\alpha=\lambda-\lambda^\alpha,\quad \lambda\in I^\alpha.  
\end{split}
\end{align}

We denote
\begin{equation*}
 c_\alpha^\pm(\lambda)=\e^{\pm \i
  \pi (d_a-3)/4}\pi ^{-1/2}\lambda_\alpha^{1/4},
\end{equation*} and $F_\rho=F(\set{x\in \bX\mid\abs{x}<\rho})$ for
$\rho>1 $ (considered below as  multiplication operators).
\begin{proposition}[\cite {Sk1}]\label{prop:radi-limits-chann22} 
  For any channel $\alpha=(a,\lambda^\alpha, u^\alpha)$,
  $\lambda\in \vE^\alpha $, $\psi\in \vB(\mathbf X)$ and $g\in \vG_a$
  there exist the  weak limits 
  \begin{align}\label{eq:restrH}
    \begin{split}
   \inp{\Gamma^\pm_{\alpha}&(\lambda)\psi,g}=\lim_{\rho\to
    \infty}\,\overline{c_\alpha^\pm(\lambda)}\rho^{-1}\\&
    \inp[\big]{F_\rho R(\lambda\pm  \i 0)\psi, F_\rho \parb{u^\alpha \otimes
    \abs{x_a}^{(1-d_a)/2}
                                      \e^{\pm \i
                                        \breve K_a(\abs{x_a}\cdot,\lambda_\alpha)}g(\pm\cdot)}}.   
    \end{split}
  \end{align}  Here    the   limits $\Gamma^\pm_{\alpha}(\lambda)$
  are  
weakly continuous $\vL(\vB (\mathbf X), \vG_a)$-valued functions of  $\lambda\in\vE^\alpha$.  
\end{proposition}

The restrictions of the map $F_\alpha (W^\pm_\alpha)^*$
have  
strong almost everywhere interpretations, meaning more precisely
\begin{align*}
  F_\alpha (W^\pm_\alpha  )^*\psi= \int^\oplus_{
  I^{\alpha} }\parb{F_\alpha (W^\pm_\alpha)^*\psi}(\lambda)\,\d
  \lambda;\quad \psi\in \vH.
\end{align*} When applied to $\psi\in \vB (\mathbf X)\subseteq
\vH $ the following relationship to Proposition  \ref{prop:radi-limits-chann22} holds. 

\begin{thm}[{\cite{Sk1}}]\label{thm:wave_matrices} For any channel
  $(a,\lambda^\alpha, u^\alpha)$  and  any $\psi\in \vB (\mathbf X)\subseteq
\vH$ 
\begin{equation}\label{eq:adjFORM}
  \parb{F_\alpha
  (W^\pm_\alpha)^*\psi}(\lambda)=\Gamma_\alpha^\pm(\lambda)\psi\quad\text{for a.e.
  } \lambda\in \vE^\alpha.
\end{equation} In particular  for any  $\psi\in \vB (\mathbf X)$ the restrictions $\parb{F_\alpha (W^\pm_\alpha)^*\psi}(\cdot)$ are
weakly continuous $\vG_a$-valued functions on  $\vE^\alpha$.  
\end{thm}

\begin{defn}\label{defn:scatEnergy07}  
An  energy $\lambda
  \in \vE:=(\min \vT(H),\infty)\setminus\vT_\p(H)$  is  {stationary
    complete}  for $H$ if  
\begin{equation}\label{eq:ScatEnergy222331}
  \forall \psi\in  L^2_\infty:\,\, \sum_{\lambda^\beta<
  \lambda}\,\norm{\Gamma_\alpha^\pm(\lambda) \psi}^2= 
  \inp{\psi,\delta(H-\lambda){\psi}}.
\end{equation} 
\end{defn} Asymptotic completeness follows by integration provided
(\ref{eq:ScatEnergy222331}) is known for almost all $\lambda \in \vE$
(motivating the used terminology). The orthogonality of the
wave operators \eqref{eq:wave_op}  implies (as demonstrated in
\cite[Subsection 9.2]{Sk1})  that
  \begin{equation}\label{eq:Besn}
     \forall \lambda
  \in \vE,\,\forall \psi\in  L^2_\infty:\,\, \sum_{\lambda^\beta<
  \lambda}\,\norm{\Gamma_\alpha^\pm(\lambda) \psi}^2\leq 
  \inp{\psi,\delta(H-\lambda){\psi}}.
  \end{equation} 

It is also known (see \cite[Proposition
9.16]{Sk1})  that a sufficient
  and necessary condition for $\lambda\in \vE$ be stationary complete is
  given as follows:

For all $\psi\in  L^2_\infty$ there exists
  $(g_\beta)_{\beta}\in \vG:=\Sigma^\oplus_{\beta }\,\vG_b$  (here 
  $\beta=(b, \lambda^\beta, u^\beta)$ runs  over all channels)  such that,  as an 
    identity in  $\vB^*/\vB_0^*$ (equipped with the quotient topology),
    \begin{equation}\label{eq:asres29}
       R(\lambda+\i 0)\psi=2\pi \i \sum_{\lambda^\beta<\lambda} J_\beta
    \breve v^{+}_{\beta,\lambda} [g_\beta],
\end{equation} where (recalling) $J_\beta \phi
=u^\beta\otimes
\phi$, here   with  $\phi$  taken to be the \emph{outgoing 
  quasi-modes} corresponding to the plus cases of
\begin{equation}\label{eq:quasiM}
 \breve v^{\pm}_{\beta,\lambda} [g]( x_b):= \mp \tfrac \i{2\pi} \parb{c_\beta^\pm(\lambda)}^{-1}\chi_+(\abs{x_b}) \abs{x_b}^{(1-n_b)/2}
                                      \e^{\pm \i
                                        \breve K_b(x_b,\lambda_\beta)}g(\pm\hat
                                           x_b);\quad g\in \vG_b. 
\end{equation} It is also known  that if
\eqref{eq:asres29} holds for some $(g_\beta)_{\lambda^\beta<\lambda}\in \vG$, then necessarily
$g_\beta=\Gamma_\beta^+(\lambda) \psi$.

The scattering matrix
$S(\lambda)=\parb{S_{\beta\alpha}(\lambda)}_{\beta\alpha}$ is given a
priori  for almost all $\lambda
  \in \vE$ by 
\begin{equation*}
  \hat
  S_{\beta\alpha}:=F_\beta(W_{\beta}^+)^*W_{\alpha}^-F_\alpha^{-1}=\int^\oplus_{
  I_{\beta\alpha} }
  S_{\beta\alpha}(\lambda)\,\d \lambda,\quad
  I_{\beta\alpha}=I^\beta\cap I^\alpha.
\end{equation*} (For $\lambda\notin I^\beta\cap I^\alpha$ we let
$S_{\beta\alpha}(\lambda)=0$.)

 The scattering matrix is known from \cite{Sk1}   to be
a  weakly continuous $\vL(\vG)$-valued function (in fact
contraction-valued) on  $\vE$. At  stationary complete energies
the scattering matrix is characterized geometrically as follows.
 \begin{thm}[{\cite{Sk1}}]\label{Cor:besov-space-setting}  Let $\lambda\in
  \vE$ be    stationary complete  and $\alpha=(a,\lambda^\alpha, u^\alpha)$ be any
  channel with $\lambda^\alpha<\lambda$. Then the following
  existence and uniqueness results hold for any  $
  g\in \vG_a$.
  \begin{enumerate}[1)]
  \item \label{item:As10} Let   $
    u=\Gamma^-_{\alpha}(\lambda)^* g$,  and let 
    $(g_\beta)_{\beta}\in \vG$ be   given by
    $g_\beta=S_{\beta\alpha}(\lambda) g$. Then,  as an
    identity in  $\vB^*/\vB_0^*$,
    \begin{align}\label{eq:as}
       u=J_\alpha
    \breve v^{-}_{\alpha,\lambda} [ g]+\sum_{\lambda^\beta<\lambda} J_\beta
    \breve v^{+}_{\beta,\lambda} [g_\beta].
    \end{align} 
  \item \label{item:As20} Conversely, if  \eqref{eq:as}
  is fulfilled for some  $ u\in \vB^*\cap H^2_{\mathrm{loc}}(\bX)$
  with $(H-\lambda)u=0$
  and for some $(g_\beta)_{\beta }\in \vG$, 
  then  $ u=\Gamma^-_{\alpha}(\lambda)^* g$ and
  $g_\beta=S_{\beta\alpha}(\lambda) g$ for all 
$\lambda^\beta<\lambda$.
\end{enumerate}
\end{thm}
\begin{thm}[{\cite{Sk1}}]\label{thm:strongC} \begin{enumerate}[1)]
  \item\label{item:1str} For any channel $\alpha$ 
    \begin{enumerate}[a)]
    \item \label{item:g} the operators
      $\Gamma^\pm_{\alpha}(\lambda)\in \vL(\vB (\mathbf X), \vG_a)$
      are strongly continuous at any stationary complete energy
      $\lambda\in\vE^\alpha$.
    \item \label{item:f} the operators
      $\Gamma^\pm_{\alpha}(\lambda)^*\in \vL( \vG_a,L^2_{-s} (\mathbf
      X))$, $s>1/2$,
      are strongly continuous in 
      $\lambda\in\vE^\alpha$.
\end{enumerate}
\item\label{item:2str} The $\vL(\vG)$-valued function
  $S(\lambda)=\parb{S_{\beta\alpha}(\lambda)}_{\beta\alpha}$  is
  strongly continuous at any stationary  complete $\lambda\in
  \vE$. Moreover the restriction of $S(\lambda)$ to the energetically
  open sector of $\vG$ is unitary at  any such energy $\lambda$.
\end{enumerate}
\end{thm}
\begin{remarks}\label{remark:R}
  \begin{enumerate}[i)]
  \item \label{item:AC1}
  For the stated results as well as for Section
  \ref{sec:Stationary completeness for the 3} the auxiliary one-body
  potential $\breve I_{a,R}$ with $R\geq 1$ taken large is a
  convenient tool. To explain its usage for the existence of the channel wave operators 
 \eqref{eq:wave_op} we introduce
\begin{align}\label{eq:brevH}
  \begin{split}
   \brh_{a,R}&=p^2_a+\brI_{a,R},\quad  \brh_{a}=\brh_{a,1},\quad
  \brH_{a,R}=H^a\otimes I+I\otimes \brh_{a,R},\\
&\brH_{a}=\brH_{a,1}\mand \brR_a(z)=(\brH_a-z)^{-1}\text{ for } z\in\C\setminus \R.
 \end{split}
\end{align} Recalling  that the channel wave operators 
 \eqref{eq:wave_op} are defined in terms of  $\breve I_{a}=\breve
 I_{a,1}$ through the definition of $\breve  S_a^\pm (p_a,t)$ we may similarly
 introduce wave operators  
\begin{align}\label{eq:wave_op30}
  \begin{split}
    \breve w_{a,R}^{\pm}&=\slim_{t\to \pm\infty}\e^{\i t\brh_{a,R}}\e^{-\i
  \breve  S_a^\pm (p_a,t)},\\
\breve W_{\alpha,R}^{\pm}&=\slim_{t\to \pm\infty}\e^{\i t\brH_{a,R}}J_\alpha\e^{-\i
  (\breve  S_a^\pm (p_a,t)+\lambda^\alpha t)}=J_\alpha \breve w_{a,R}^{\pm}.
  \end{split}
\end{align}  Now the well-definedness of  the channel wave operators 
 \eqref{eq:wave_op} follows by combining \cite{Sk1} with the constructions
 $\brH_{a,R}$. Indeed to show the existence of the limits
 \begin{equation*}
   \lim_{t\to \pm\infty}\e^{\i tH}J_\alpha\e^{-\i
  (\breve  S_a^\pm (p_a,t)+\lambda^\alpha t)}\varphi=\lim_{t\to
  \pm\infty}\e^{\i tH} \e^{-\i t \brH_{a,R}}J_\alpha\breve w_{a,R}^{\pm}\varphi,
 \end{equation*} it suffices to consider $\varphi$ localized as 
 $\varphi=f_1(k_\alpha)\varphi$, where $f_1$ is any standard support
 function supported near a
 fixed $\lambda_0\notin \vT_{\p}(H)$. Furthermore we can assume that
 the Fourier transform $\hat \varphi$ is supported away from
 collision planes, and with this assumption the proof reduces to the
 existence of the limits
 \begin{subequations}
 \begin{equation}\label{eq:aux}
   \slim_{t\to \pm\infty}\,\e^{\i tH}
  \Phi_{a,R}^\pm  \e^{-\i t\brH_{a,R}},
 \end{equation} where
\begin{equation}\label{eq:aux2}
  \Phi_{a,R}^\pm={f_2}(H)M_a N^a_\pm M_a
   {f_2}(\brH_{a,R}),\,  f_2\succ f_1.
\end{equation}  
 \end{subequations}   Here  the factors $M_a$ and $ N^a_\pm$ are suitable `localization operators' from 
\cite[Section 3]{Sk1}, to be elaborated on  in Subsections 
\ref{subsec:Yafaev's constructions} and \ref{subsec:A phase-space partition of unity}. There are Mourre estimates for $H$ as well as for $\brH_{a,R}$
at $\lambda_0$ provided   $R\geq 1$ is chosen large
enough  (see  \cite[Subsection 5.1]{Sk1} for details). Consequently  if $f_2$ is
also narrowly supported,
then  the procedure of \cite{Sk1} yields the existence  of \eqref{eq:aux}.
\item\label{item:AC20} The existence of the limits \eqref{eq:aux}
  relies on Kato-smoothness bounds obtained by concrete commutator
  bounds. A related technique, on which Section \ref{sec:Stationary
    completeness for the 3} of the present paper will be based, concerns   `$Q$-bounds' of
  the resolvent. These take the form 
\begin{equation}\label{eq:2boundobtain3300}
  \sup _{\Im z\neq 0}\norm[\big]{\abs{Q{f_1}
   (H)}{R(z)}}_{\vL(\vB,\vH)}< \infty.
\end{equation} In most cases the relevant $Q$-bounds are derived by
computing $\i [H,\Psi]$ for a good choice of a bounded self-adjoint operator $\Psi$
(a
'propagation observable') and then extracting its (dominating)
positive part. See for example \cite[Lemma 2.2]{Sk2} for a precise
assertion  (this result  also appears as
\cite[Lemma B.1]{Sk1}). We have stated   complete lists  of
$Q$-bounds needed in the paper (for $H$ as well as for auxiliary
Hamiltonians) in (\ref{eq:kato10}), \eqref{eq:phiweak}
and (\ref{eq:phiweak22}).

\item \label{item:AC2} In \cite{Sk1} (and above) it is convenient to
  use an 'extended lattice structure' (which include the spaces
  $\bX^a$ as collision planes) rather than $\vA$, see
  \cite[Subsection 3.1]{Sk1}. This allows us to consider the operators
  $\brH_{a,R}$ on an equal footing with the Hamiltonian $H$, in
  particular there are Mourre estimates for these operators, and
  indeed in Subsection \ref{subsec:Conclusion and generalizations} we
  take $R$ large to assure a Mourre estimate at the given energy of
  interest. 

Although  the operators $M_a$ in \cite{Sk1} are  constructed from
the extended lattice structure this setup  
    is
  not appropriate for our  analysis in   Section
  \ref{sec:Stationary completeness for the 3}. Rather  the  `channel localization operators' $M_a$ from
Subsection \ref{subsec:Yafaev's constructions} are  constructed from the original lattice
  structure. However, still a  different  type  of
  lattice  structure will be needed as a technical tool for treating
  some commutators. This third structure takes the following form: For
  any $a\in \vA_1$ we consider  the family $\set{\bX, \set{0}, \bX_a,
    \bX^a}$ as the collection of collision planes (defining a unique
  structure). For any such $a$ this structure works well for deriving
  $Q$-bounds of the resolvent of $\brH_{a}$, as demonstrated in Subsection \ref{subsec:Yafaev's constructions}.

\item \label{item:AC3
} Finally we note that all the stated results in this section differ
slightly from their origin  \cite{Sk1} in that the above one-body
wave operators $\breve w_{a,R}^{\pm}$ appeared in a slightly different
form in   \cite{Sk1} (the ones there involved a solution to the
Hamilton-Jacobi equation for the potential $ I_{a,R}$ rather than for
$ I_{a,1}$ as above). As a result our account on \cite{Sk1} presented
in this subsection  appears
slightly  cleaner, we think. The relationship between the wave
operators \eqref{eq:Atomwave_op} and \eqref{eq:wave_op} is explained
in \cite[Remarks 2.2]{Sk1}. In particular the results in Subsection
\ref{subsec: Atomic 3-body model} follow from those  presented above and Theorem
\ref{them:stat-compl-enerMain} stated below.
\end{enumerate} 
\end{remarks}
\section{Stationary complete energies  for the $3$-body problem}\label{sec:Stationary completeness for the 3}

 The main result of the paper reads. 
\begin{thm}\label{them:stat-compl-enerMain} Suppose  Condition \ref{cond:smooth2wea3n12} and the
$3$-body condition  \eqref{eq:3body}.
  Then all $\lambda
  \in \vE=(\min \vT(H),\infty)\setminus {\vT_{\p}(H)}$ are stationary
  complete for $H$.
\end{thm}
 To prove this result we first fix any 
$\lambda_0>0$ and  impose   the following (simplifying) condition for
all $a\in \vA$:
\begin{equation}\label{eq:sim2} H^a \text{ does not  have positive eigenvalues}.
  \end{equation}
  Note that although indeed for a big class of potentials
  $\vT_\p(H)\cap \,\R_+=\emptyset$, cf. \cite{AIIS, FH}, the property
  \eqref{eq:sim2}  is not known under \eqref{eq:3body} and Condition
  \ref{cond:smooth2wea3n12}.

  We are first going to derive asymptotics of $\phi=R(\lambda+\i
  0)\psi$ in agreement with \eqref{eq:asres29} for $\lambda=\lambda_0$
  and for any $\psi\in L^2_\infty$ (henceforth fixed) yielding the
  desired completeness assertion for this $\lambda$ under
  \eqref{eq:sim2}. This task will occupy Subsections \ref{subsec:Yafaev's constructions}--\ref{subsec:Free channel
    contribution}. We devote then Subsection \ref{subsec:Conclusion and
    generalizations}  to doing the general case by adding to the
  previous pattern of proof an elementary   cut-off scaling argument.

 Recall the notation  $a_0=a_{\max}$,   $d= \dim \bX$ and $d_a= \dim \bX_a$.

\subsection{Yafaev's constructions and some $Q$-bounds}\label{subsec:Yafaev's
  constructions}

We need  to consider various conical
subsets of  $\mathbf X\setminus{0}$.
 Let for  $a\in\vA_1$ and $\varepsilon,\delta\in (0,1)$
\begin{align}\label{eq:primes}
  \begin{split}
    \mathbf X'_a&={\mathbf X_a }\setminus\cup_{{  b\gneq a,\,b\in \vA} }\,\mathbf
                 X_b={\mathbf X_a }\setminus\cup_{{b\not\leq a,\,b\in
                     \vA}\,
                 }\mathbf X_b,\\
\mathbf X_a(\varepsilon)&=\set{x\in \bX\mid\abs{x_a}>(1-\varepsilon)\abs{x}},\\
\mathbf \Gamma_a(\varepsilon)&=\parb{\mathbf X \setminus\set{0}}\setminus\cup_{{b\not\leq a,\,b\in \vA_1}
}\,\mathbf X_b(\varepsilon),\\
\mathbf Y_a(\delta)&=\mathbf X_a(\delta)\setminus\cup_{b\gneq a,\,b\in \vA_1}
\,\overline{\mathbf X_b(3\delta^{1/d_a})}.
\end{split}
\end{align} Here  the overline means topological closure in
$\mathbf X$. 
The structure of the sets $\mathbf X_a(\varepsilon)$, $\mathbf
\Gamma_a(\varepsilon)$ and $\mathbf Y_a(\varepsilon)$ is
${\R_+ V}$, where  $V$ is a
subset of the unit sphere $\S^{d-1}$ in
$\mathbf X$. For $\mathbf
X_a(\varepsilon)$ and $\mathbf Y_a(\varepsilon)$ the set $V$ 
 is relatively open, while for  $\mathbf
\Gamma_a(\varepsilon)$ the  set is  compact.

We also note that 
 \begin{equation}\label{eq:Gamma2}
   \forall a\in\vA_1\, \forall \varepsilon\in (0,1);\q \mathbf \Gamma_a(\varepsilon)\subseteq\cup_{b\leq a}  \,\mathbf X'_b.
\end{equation} This is a very elementary property under the  three-body condition
  \eqref{eq:3body} (for the  general case, see for example  \cite[Lemma 3.10]{Sk2}).

 Thanks to   \eqref{eq:Gamma2} we can for  any $a\in\vA_1$ and  any  $\varepsilon, \delta_0\in (0,1)$  write
  \begin{equation}\label{eq:deltaNBHa}
    \mathbf \Gamma_a(\varepsilon)\subseteq\cup_{b\leq a} \cup_{\delta\in (0,\delta_0]} \,\mathbf Y_b(\delta).
  \end{equation}

      As the reader will see later we will use (\ref{eq:deltaNBHa}) 
     with $\varepsilon=\epsilon^d$,   where $\epsilon>0$ is a small
      parameter in terms of which the Yafaev functions $m_a$ (for 
      $a\in\vA_1$ as well as for $a=a_0$) all depend from the very
      construction, see \cite[Subsection 3.1]{Sk2}.  

We let
for this parameter  $\epsilon>0$ and  $a\in \vA_1$ 
\begin{equation*}
  \varepsilon^a_k= k\epsilon^{d_a};\quad k=1,2,3,4.
\end{equation*}

\begin{lemma}[\cite{Sk2}]\label{lemma:ma1}
For any $a\in\vA_1$ the  function 
 $m_a:\mathbf X\setminus \set{0}\to \R$ (depending on  a  sufficiently
 small  parameter $\epsilon>0$) fulfils  the following
properties for any $b\in\vA_1$: 
\begin{enumerate}[1)]
\item\label{item:10a} $m_a$ is homogeneous of degree $1$.
\item\label{item:11a} $m_a\in C^\infty (\bX\setminus\set{0})$.
\item\label{item:12a} If $b\leq a$  and  $x\in \mathbf X_b(\varepsilon^b_1)$, then
  $m_a(x)=m_a(x_b)$.
\item\label{item:13a} If ${b\not\leq a}$ and $x\in \mathbf
  X_b(\varepsilon^b_1)$, then   $m_a(x)=0$.
\item\label{item:14a} If $a\neq {a_{\min}}$,  $x\in \bX\setminus
  \set{0}$  and $x\notin \bX_a(\varepsilon^a_3)$ (i.e. that $\abs{x^{a}}\geq \sqrt{\ \varepsilon^a_3(2-\varepsilon^a_3)}\abs{x}$), then   $m_a(x)=0$.
\end{enumerate} 
\end{lemma} 

\begin{lemma}[\cite{Sk2}]\label{lemma:m1} The  function 
 $m_{a_0}:\mathbf X\setminus  \set{0}\to \R$ (depending on  a  sufficiently
 small  parameter $\epsilon>0$) fulfils  the following
properties: 
\begin{enumerate}[i)]
\item\label{item:10b} $m_{a_0}$ is convex and homogeneous of degree $1$.
\item\label{item:11b} $m_{a_0}\in C^\infty (\mathbf X\setminus  \set{0})$.
\item\label{item:12b} If $b\in\vA_1$  and  $x\in \mathbf X_b(\varepsilon^b_1)$, then
  $m_{a_0}(x)=m_{a_0}(x_b)$.
\item\label{item:9b} $m_{a_0}=\Sigma_{a\in \vA_1}\,m_a$.
\item\label{item:13b} For any  $a\in \vA_1$ there exists  $c_a\geq
  1$  (depending on the
parameter  $\epsilon$): If $x\in\mathbf
  X_a(\varepsilon^a_1)$ obeys that for all  $b\in \vA_1$
  with  $b\gneq
  a $ the vector $x\not\in\mathbf
  X_b(\varepsilon^b_3)$, then 
  \begin{equation}\label{eq:goodEps}
    m_{a_0}(x)=m_a(x)=c_a \abs{ x_{a}}.
  \end{equation}
\item\label{item:15b} There exists $C>0$ (being independent of the
parameter  $\epsilon$) such that for all $x\in
X\setminus  \set{0}$ 
\begin{equation}
  \label{eq:compa}
  \abs{\nabla \parb{m_{a_0}(x)-\abs{x}}}\leq C\sqrt{\epsilon}.
\end{equation}
\end{enumerate} 
\end{lemma}

 The functions $m_a$ from Lemma \ref{lemma:ma1}
and $m_{a_0}$ from Lemma \ref{lemma:m1} lack  smoothness at 
$0\in \bX$. This deficiency is cured by
multiplying them by a suitable factor, say specifically by the factor
$\chi_+(2|x|)$. We adapt in the following these smooth modifications and will use (slightly
abusively) the same notation $m_a$ and $m_{a_0}$ for the smoothed out
versions of the Yafaev functions.  We may then consider the corresponding first order
operators $M_a$ (including  $M_{a_0}$) realized as self-adjoint
operators
\begin{equation}\label{eq:M_a0}
 M_a=2\Re(w_a\cdot p)=-\i\sum_{j\leq d}\parb{(w_a)_j\partial_{x_j}+\partial_{x_j}(w_a)_j};\quad  w_a=\mathop{\mathrm{grad}} m_a.   
\end{equation}
The operators $M_a$, $a\in\vA_1$, may  and will be be considered as `channel
localization operators', while  operators of the form $M_{a_0}$ (with
adjusted values of the parameter $\epsilon$)
primarily will enter as a technical quantities  controlling  commutators of  
the Hamiltonian and the channel
localization operators. This matter will be elaborated on below and further studied 
in the subsequent subsections. 

\subsubsection{$Q$-bounds for $H$}\label{subsubsec:Q-bounds for H} 
We will complete the present subsection by proving
various `$Q$-bounds'  to control the commutators
$\i[H,M_a]$. We connect (\ref{eq:deltaNBHa})  and Lemmas
\ref{lemma:ma1} and \ref{lemma:m1}. Although the (small) positive
parameter $\epsilon$ of 
Lemmas \ref{lemma:ma1} and \ref{lemma:m1}  can be chosen independently
we first
choose and fix the same small $\epsilon$ for the lemmas. (This
particular $M_{a_0}$  will be used in \eqref{eq:small9}.)
 
 Thanks to 
 Lemma  \ref{lemma:ma1}
\ref{item:13a} we  can record that  
\begin{equation}\label{eq:suppa}
  \supp m_a\subseteq 
\mathbf \Gamma_a(\epsilon^{d});\q a\in\vA_1.
\end{equation} Hence with  $\varepsilon=\delta_0:=\epsilon^{d}$ 
in \eqref{eq:deltaNBHa}  we obviously have obtained a covering of the
support of $m_a$. 
It turns out to be convenient to use a slightly refined covering, more
precisely given in terms of the sets
\begin{equation*}
   \bY'_{b}(\delta):= \bY_{b}(\delta)\cap {\bX_a(\varepsilon^a_4)};\q
   b\leq a,\,\delta \leq \delta_0.
\end{equation*} Thanks to  Lemma  \ref{lemma:ma1}
\ref{item:14a} it follows that 
\begin{equation*}
  \supp m_a  \subseteq\cup_{b\leq a} \cup_{\delta\in (0,\delta_0]} \,\mathbf Y'_b(\delta).
\end{equation*}
By compactness  we can choose $\delta_1,\dots, \delta_J\in
(0,\delta_0]$ and  $a_1,\dots, a_J\leq a$  such that
\begin{equation}\label{eq:suppa9}
  \supp m_a\subseteq 
\cup_{j\leq J} \,\,\mathbf Y'_{a_j}(\delta_j);\q J=J(a)\in \N.
\end{equation} 

We can 
simplify  (\ref{eq:suppa9})  under the  three-body condition
 \eqref{eq:3body} as 
\begin{align}\label{eq:3inclusion}
    \supp m_a\subseteq \begin{cases}\mathbf Y'_{a_{\min}} (\delta_{a_{\min}}),\q 
&\text{if } a=a_{\min},\\
\mathbf Y'_{a} (\delta_{a})\cup \mathbf Y'_{a_{\min}}
  (\delta^a_{a_{\min}}),\q 
&\text{if } a\in \vA_2.
    \end{cases}
  \end{align} 
  Here we may take $\delta_{a}=
  \delta_0=\epsilon^{d}$ for $a\in \vA_2$ in which case $\mathbf Y_{a} (\delta_{a})=\mathbf X_{a}
  (\delta_{a})$, while 
$\delta_{a_{\min}}, \delta^a_{a_{\min}}\leq
  \epsilon^{d}$ need to be taken 
  smaller. In particular for each $a\in\vA_1$ the corresponding
  $J$ in (\ref{eq:suppa9}) is either one or two. We prefer to use
  the uniform notation of (\ref{eq:suppa9}) rather than the more
  cumbersome notation of (\ref{eq:3inclusion}).

Now we need applications of Lemma \ref{lemma:m1} for
$\epsilon_1,\dots,\epsilon_J\leq \epsilon$ fixed as follows. Since
$\delta_j\leq \delta_0$, we can introduce positive
$\epsilon_1,\dots,\epsilon_J\leq \epsilon$ by the requirement
$\epsilon_j^{d_{a_j}}=\delta_j$. The inputs $\epsilon=\epsilon_j$ in
Lemma \ref{lemma:m1} yield corresponding functions, say denoted
$m_j$. In particular in the region $\bY_{a_j}(\delta_j)$ the function
$m_a$ from Lemma \ref{lemma:ma1} only depends on $x_{b_j}$ (thanks to
Lemma \ref{lemma:ma1} \ref{item:12a} and the property
$\bX_{a_j}(\delta_j)\subseteq \bX_{a_j}(\epsilon^{d_{a_j}})$), while
(thanks to Lemma \ref{lemma:m1} \ref{item:13b})
    \begin{subequations}
    \begin{equation}\label{eq:good0}
      m_j(x)=c_j\abs{x_{a_j}},\quad c_j=c_{a_j}.
    \end{equation}  Obviously $\abs{y'}$
    is
    non-degenerately  convex in $y'\in\bX_{a_j}\setminus \set{0}$, meaning that  the  restricted Hessian
    \begin{equation}\label{eq:HesRestri0}
      \parb{\nabla_{y'}^2
        \abs{y'}}_{|\bX_{a_j}\cap\set{y}^\perp}\text{ is
      positive definite at any }y\in \bX_{a_j}\setminus\set{0}.
    \end{equation}   
    \end{subequations}

These properties   can be applied as follows  using   
for  $b\in\vA_1$  the vector-valued first order operators 
\begin{equation}\label{eq:Ffield}
  G_b=\vH_b(x_b)\cdot
  p_b,\text{ where } \vH_b(x_b)=\chi_+(2\abs{x_b})\abs{x_b}^{-1/2}\parb{I-\abs{x_b}^{-2}\ket{ x_b}\bra{ x_\textbf{}}}.
\end{equation}  
 We  choose for the considered   $a\in\vA_1$ 
   a quadratic partition $\xi_1,\dots, \xi_J\in
C^\infty(\S^{d-1})$ (viz $\Sigma_j \,\xi_j^2=1$)
subordinate to the covering \eqref{eq:suppa9}. Then
we can write 
\begin{align*}
  m_a(x)=\Sigma_{j\leq J}\,\,m_{a,j}(x);\quad m_{a,j}(x)=\xi^2_j( \hat x) m_a(x),\quad \hat x=x/\abs{x},
\end{align*} and from the previous discussion it follows  that 
\begin{align*}
  \chi^2_+(\abs{x})m_{a,j}(x)=\xi_j^2( \hat x)\chi^2_+(\abs{x}) m_{a}(x_{a_j}),
\end{align*} as well as 
\begin{subequations}
\begin{align}\label{eq:Hes0}
  \begin{split}
 &p\cdot\parb{\chi^2_+(\abs{x})\nabla^2m_a(x)} p=\Sigma_{j\leq J}\,\,
                                 p\cdot\parb{\xi^2_j( \hat x)\chi^2_+(\abs{x}) \nabla^2m_a(x_{a_j})}p,\\&
=\Sigma_{j\leq J}\, \,G^*_{b_j}\parb{\xi^2_j( \hat x)\chi^2_+(\abs{x})\vG_j}G_{b_j};\quad 
                                                                  \vG_j=\vG_j(x_{a_j})\text{ 
                                                                  bounded}.   
  \end{split}
\end{align}
In turn using the convexity property of $m_j$ and the previous
discussion (cf. \eqref{eq:good0} and
\eqref{eq:HesRestri0})  we deduce the bound  
\begin{align}\label{eq:Hes_est}
   G^*_{a_j}\xi^2_j( \hat x)\chi^2_+(\abs{x})
  G_{a_j}\leq  2 p\cdot \parb{\chi^2_+(\abs{x})\nabla^2m_j(x)} p.
\end{align}

Here the right-hand side  is the `leading term' of  the commutator
$\tfrac 12 \i[p^2,M_j]$, where $M_j$ is given by \eqref{eq:M_a0} for  the
modification of 
$m_j$ given by the  function $\chi_+(2|x|)m_j(x)$. More precisely  for any real $f\in C_\c^\infty(\R)$
\begin{align}\label{eq:1712022apll}
	\begin{split}
		f(H)\parbb{2 p\cdot \parb{\chi^2_+(\abs{x})\nabla^2m_j(x)} p-\tfrac 12 \i[
H,&\chi_+(\abs{x})M_j\chi_+(\abs{x})] }f(H) \\&=\vO(\inp{x}^{-1-\mu}).
\end{split}
\end{align}  
\end{subequations}

Similarly  the leading term of  the commutator
$\i[H,M_a]$  is given as 
\begin{align}\label{eq:MaLeading1}
	\begin{split}
		&f(H) \i[
H,M_a] f(H) \\&=4f(H) p\cdot \parb{\chi^2_+(\abs{x})\nabla^2m_a(x)}
pf(H) +\vO(\inp{x}^{-1-\mu})
\\&=4\sum_{j\leq J}\, \,f(H) Q_j^*\vG_jQ_jf(H)
+\vO(\inp{x}^{-1-\mu});\q Q_j:=\xi_j( \hat x)\chi_+(\abs{x})G_{a_j}.
\end{split}
\end{align}

Next we combine the features 
\eqref{eq:Hes0}--\eqref{eq:1712022apll}  with \cite[Lemma 2.2]{Sk2},
the latter  applied concretely with the propagation observable 
\begin{subequations}
\begin{equation}\label{eq:PropBasic}
  \Psi=\Psi_j=\tfrac 12 f_1(H)\chi_+(\abs{x})M_j\chi_+(\abs{x})f_1(H),\q M_j=M_j(a),
\end{equation} where 
$ f_1$ is any  narrowly supported standard support
function obeying  $f_1=1$ in a
neighbourhood of a given $\lambda\not\in \vT_{\p}(H)$, cf. Remark
\ref{remark:R} \ref{item:AC20}. A commutator calculation (using the familiar  Helffer--Sj\"ostrand
formula, see \eqref{82a0} stated below) leads
to the basic `$Q$-bound' of \cite[Lemma 2.2]{Sk2} with $Qf_1 (H)$
obeying 
\begin{equation}\label{eq:HessionPos}
  {\abs{Qf_1 (H)}}^2=2f_1(H){ p\cdot \parb{\chi^2_+(\abs{x})\nabla^2m_j(x)} p }f_1(H),
\end{equation} cf.  \cite[(3.28b)]{Sk2},  and therefore in turn  to the $Q$-bounds 
\begin{align}\label{eq:2boundobtain33}
  \begin{split}
  \sup _{\Im z\neq 0}\norm{Q(a,j){f_1}&
   (H){R(z)}}_{\vL(\vB,\vH^{d})}< \infty;\\
&\quad Q(a,j)=\xi_j( \hat x)\chi_+(\abs{x})G_{a_j},\,\,j\le J=J(a).
 \end{split}
\end{align} 
\end{subequations}

We introduce for $a\in\vA_1$ functions $\xi^+_a$ and $\tilde\xi^+_a$ as follows.
First choose any $\xi_a\in
C^\infty(\S^{d-1})$  such that $\xi_a=1$ in
$\S^{d-1}\cap\mathbf  \Gamma_a(\varepsilon)$ and $\xi_a=0$ on
$\S^{d-1}\setminus  \mathbf \Gamma_a(\varepsilon/2)$. Choose then any $\tilde\xi_a\in
C^\infty(\S^{d-1})$ using this recipe   with $\varepsilon$
replaced by  $\varepsilon/2$. Finally  let  $\xi^+_a(x)=\xi_a(\hat
x)\chi_+(4\abs{x})$ and  $\tilde\xi^+_a(x)=\tilde\xi_a(\hat
x)\chi_+(8\abs{x})$,  and note  that
$\tilde\xi^+_a\xi^+_a=\xi^+_a$. Applied to $\varepsilon=\epsilon^{d}$ it follows from Lemma
 \ref{lemma:ma1} \ref{item:13a} that the channel localization operators $M_a$ fulfil   
\begin{align}\label{eq:partM}
  M_a= M_a\xi^+_a= \xi^+_aM_a=M_a\tilde\xi^+_a=\tilde \xi^+_aM_a,
\end{align} which in applications  provides `free factors' of
$\xi^+_a$ and $\tilde\xi^+_a$
where  convenient. In particular \eqref{eq:partM}  will  be  useful (under conditions and by
commutation)  
for replacing $H$ by $\brH_a$ (or vice versa) in the presence of a
factor $M_a$.

\subsubsection{$Q$-bounds for $\brH_a$}\label{subsubsec:Q-bounds for
  breveH} 
 Finally  we discuss  variations  of \eqref{eq:2boundobtain33}  for
 $\brH_a$, $a\in \vA_1$, rather than  for $H$. This depends on the
 same 
 covering (\ref{eq:suppa9})  and the same quadratic partition $\xi_1,\dots, \xi_J\in
C^\infty(\S^{d-1})$  
subordinate to the covering. We choose then again corresponding convex
functions $m_j$, however taken slightly differently: Now the construction for the
given parameter $\epsilon_j$ is based on the lattice structure
$\set{\bX, \set{0}, \bX_{a}, \bX^{a}}$ of collision planes rather than
the old one parametrized by $\vA$. In this way the commutator $\i[\brH_a,M_j]$ is as `good' as $\i[H,M_j]$  with the
old $M_j=M_j(a)$ (discussed above). This is due to the fact that $\bX^{a}$ now is considered as
a collision plane (making  the contribution from both of the potentials $V^a$ and $\brI_{a,1}$  well controlled). Although globally the old and the new $M_j$'s in
general are different, they coincide on the support of $\xi_j$ by the
proof of Lemma \ref{lemma:m1} (not repeated in this paper), so indeed the previous arguments work
and we can conclude the $Q$-bounds
\begin{align}\label{eq:2boundobtain3350}
  \begin{split}
  \sup _{\Im z\neq 0}\norm{Q(a,j){f_1}&
   (\brH_a){\brR_a(z)}}_{\vL(\vB,\vH^{d})}< \infty;\\
&\quad Q(a,j)=\xi_j( \hat x)\chi_+(\abs{x})G_{a_j},\,\,j\le J=J(a).
 \end{split}
\end{align}  The bound (\ref{eq:2boundobtain3350}) will be used below
and in the subsequent subsections to treat the commutator
$\i[\brH_a,M_a]$. Let us here note the following analogue of \eqref{eq:MaLeading1}
\begin{align}\label{eq:MaLeading2}
	\begin{split}
		&f(\brH_a) \i[
\brH_a,M_a] f(\brH_a) \\&=4f(\brH_a) p\cdot \parb{\chi^2_+(\abs{x})\nabla^2m_a(x)}
pf(\brH_a) +\vO(\inp{x}^{-1-\mu})
\\&=4\sum_{j\leq J}\, \,f(\brH_a) Q_j^*\vG_jQ_jf(\brH_a)
+\vO(\inp{x}^{-1-\mu});\q Q_j:=\xi_j( \hat x)\chi_+(\abs{x})G_{a_j}.
\end{split}
\end{align}  

\begin{subequations}
By repeating the analysis for  $b\in \vA_1$, $b\neq a$, using now the propagation observable
\begin{equation}\label{eq:PropBasic60}
  \Psi=\Psi_j=\tfrac 12
  f_1(\brH_a)M_a\chi_+(\abs{x})M_j\chi_+(\abs{x})M_af_1(\brH_a),\q M_j=M_j(b),
\end{equation} in combination with (\ref{eq:2boundobtain3350}) we can
then  deduce the somewhat similar 
 $Q$-bounds
\begin{align}\label{eq:2boundobtain33500}
  \begin{split}
  \sup _{\Im z\neq 0}\norm{Q_a(b,j){f_1}&
   (\brH_a){\brR_a(z)}}_{\vL(\vB,\vH^{d})}< \infty;\\
&\quad Q_a(b,j)=\xi_j( \hat x)\chi_+(\abs{x})G_{b_j}M_a,\,\,j\le J=J(b).
 \end{split}
\end{align}   
 Here $b_j\leq b$ and the functions $\xi_j$ are quadratic partition functions,
not relative to \eqref{eq:3inclusion}, but for 
 the  covering 
\begin{align}\label{eq:3inclusionsimp150}
    \supp m_b\subseteq \begin{cases}\mathbf Y'_{a_{\min}} (\delta_{a_{\min}}),\q 
&\text{if } b=a_{\min},\\
\mathbf Y'_{b} (\delta_{b})\cup\mathbf Y'_{a_{\min}}
  (\delta^b_{a_{\min}}),\q 
&\text{if } b\in \vA_2;
    \end{cases}
  \end{align} alternatively and more conveniently denoted in  the same
  way as before,  i.e. as  
\begin{equation}\label{eq:suppa950}
  \supp m_b\subseteq 
\cup_{j\leq J} \,\,\mathbf Y'_{b_j}(\delta_j).
\end{equation}
\end{subequations} The convex functions $m_j=m_j(b)$ used in
\eqref{eq:PropBasic60} are constructed from the original lattice
structure $\vA$  (as  in \eqref{eq:PropBasic} for $a$, but now for   $b$).

Note the appearance of the factor $M_a$ in
(\ref{eq:2boundobtain33500}). In practice, although $\i[\brH_a,M_j(b)]$
may not be treatable standing alone (in constrast to the case $b=a$
discussed above) the expression $M_a\i[\brH_a,M_j(b)]M_a$ is `good'
thanks to the support properties \eqref{eq:3inclusion} and
\eqref{eq:3inclusionsimp150}.  The above commutator argument used
for (\ref{eq:2boundobtain33500}) obviously depends on
\eqref{eq:MaLeading2}, the previous
$Q$-bounds \eqref{eq:2boundobtain3350} and various commutation. 

\subsection{A phase-space   partition of unity}\label{subsec:A phase-space partition of unity}
 We will in this
subsection use the `channel localization operators' $M_a$ from
Subsection \ref{subsec:Yafaev's constructions}  (used also in \eqref{eq:aux2}) to construct a certain
`effective partition of unity', say denoted $\Sigma_{a\in
  \vA_1}\,S_a\approx I$, which roughly will allow us to reduce the problem of
asymptotics to that of $\brR_a(\lambda+\i0)\psi$, $a\in \vA_1$.  For
that purpose some properties from \cite{Sk2} (not derived in the
seminal paper \cite{Ya3}) are needed.

We recall that 
  $a_0=a_{\max}$ and consider the corresponding operator
  $M_{a_0}=\Sigma_{a\in \vA_1}\,M_{a}$. We recall that the construction of the channel
localization operators $M_a=2\Re (p\cdot\nabla m_a)$ depends on a sufficiently small parameter 
$\epsilon>0$ (as before, in the following
  mostly  suppressed). In particular one can roughly think of
  $m_{a_0}=\Sigma_{a\in \vA_1} m_a\approx |x|$.

  For   $a\in \vA_2$  we introduce  operators $N^a$  of
the form 
\begin{align}\label{eq:propgaObs}
  \begin{split}
  N^a&=A_1 A_2^a(A_3^a)^2A_2^aA_1;\\
A_1&=\chi_+(
  B/\epsilon_0),\\
A_2^a &=\chi_-\parb{
  r^{\rho_2-1}r_\delta^a},\\
A_3^a&=\chi_-\parb{
   r^{\rho_1/2}B_\delta^ar^{\rho_1/2}},
  \end{split}
\end{align} where  $\epsilon_0>0$ is sufficiently  small as 
  determined by Mourre estimates at $\lambda$, cf.  Remarks
 \ref{remark:R}, 
\begin{align}\label{eq:parameters}
  1-\mu<\rho_2<\rho_1<1-\delta\text{ with } \delta\in (2/(2+\mu),\mu),
\end{align}
 and
$ r, B, r_\delta^a$ and $B_\delta^a$ are operators constructed by quantities from
\cite{De} ($r$ and $r_\delta^a$ are multiplication operators while
$B$ and $B_\delta^a$ are corresponding Graf vector field type
constructions).    
 More precisely 
 $r  $ is  a positive smooth function on $\mathbf X$ which,  apart from a trivial rescaling to assure
 Mourre estimates for the  Graf
vector field $\nabla r^2/2$  (as done in \cite [Subsection 5.1]{Sk1}),
is taken as   the function $r$ constructed in  \cite{De} (roughly one
can think of $r$ as 
$r(x)\approx |x|$ like the above function $m_{a_0}$, although their finer
properties are very different).  This  function $r$  partly plays the
role  of  a
`stationary time variable' compared to the usage of the real time
parameter in  \cite{De}. (It should not be mixed up with the function
$\abs{x}$.) The operator $B:=2\Re(p\cdot\nabla r)$. Let $r^a$ be  the same
function now constructed on $\mathbf X^a$ rather than on $\mathbf X$. Let then $r_\delta^a=
r^\delta r^a(x^a/r^\delta)$ and 
$B_\delta^a=2\Re\parb{p^a\cdot(\nabla r^a)(x^a/r^\delta)}$. For $a= a_{\min}$ we  take $
N^a=A_1^2$.

We recall (see  \cite [Subsection 5.2]{Sk1})  that for the above  (small)
$\epsilon_0>0$
    \begin{equation}\label{eq:MicroB}
      \chi_-( \pm B/\epsilon_0) R(\lambda\pm \i0)\psi'\in \vB_0^*
      \text{ for all }\psi'\in \vB.
    \end{equation}  At this point it should be noted that
    $\chi_-(  \pm B/\epsilon_0)\in \vL(\vB)$, cf.  \cite[Theorem
    14.1.4]{H1}. 
\begin{remarks}\label{remark:ExtndeNot}
  \begin{enumerate}[i)]
  \item \label{item:LAPext1} Thanks to \eqref{eq:thre00}, \eqref{eq:ppH} and
    \eqref{eq:thre} stated below \eqref{eq:MicroB} is also valid with
    $H$ replaced by the operator $\brH_a$ from \eqref{eq:brevH} (for
    which also $\brR_a(\lambda+\i 0)$ makes sense). See
    \eqref{eq:strongED_+SDY} for a related estimate. For these
    properties an extended lattice structure is used in the
    construction of the above operator $B$, cf. Remark \ref{remark:R}
    \ref{item:AC2}. 
\item \label{item:LAPext2} The construction of the above operators
  $M_a$ depends on the original  lattice structure $\vA$ only, cf. Remark \ref{remark:R}
    \ref{item:AC2}. 
\end{enumerate}
\end{remarks}

  Thanks to \eqref{eq:MicroB}  our problem is  reduced to  the
  asymptotics of $ \chi_+( B/\epsilon_0)R(\lambda+\i0)\psi$. It  is
  further reduced thanks to the following elementary estimate \eqref{eq:small}, cf.  \cite[Lemma
  5.1]{Sk2} and its proof. Recall the concept of order $t\in \R$ of an operator $T$
  as defined in Subsection \ref{subsec:body Hamiltonians, limiting absorption
  principle  and notation} and there  expressed as
$T=\vO(\inp{x}^{t})$. Recall also the   Helffer--Sj\"ostrand
formula, assuming here that $T$ is self-adjoint and that $f\in C ^\infty(\R)$,
\begin{align}\label{82a0}
  \begin{split}
 f(T) 
&=
\int _{\C}(T -z)^{-1}\,\mathrm d\mu_f(z)\, \text{ with}\\
&\mathrm d\mu_f(z)=\pi^{-1}(\bar\partial\tilde f)(z)\,\mathrm du\mathrm dv;\quad 
z=u+\i v.   
  \end{split}
\end{align} The function  $\tilde{f}$ is an `almost analytic' extension of
$f$, which in this case may be taken compactly supported. However the 
formula \eqref{82a0} extends to more general  classes of  functions
and  serves as a standard tool for commuting  operators. Since  it
will be used only tacitly in this paper  the interested  reader might benefit from
consulting 
 \cite[Section 6]{Sk1} which  is devoted to
applications of \eqref{82a0} to  $N$-body Schr\"odinger operators,
hence being equally relevant for the present paper. Finally   recall the generic
  notation $\inp{T}_{\varphi}=\inp{\varphi,T\varphi}$.

\begin{lemma}\label{lem:compBogB} Let $f_1,f_2$ and $f_3$ be standard
  support functions with 
$f_3\succ f_2\succ f_1 $.
 Let
 \begin{equation}\label{eq:tildeM}
   M=\Sigma_{a\in
   \vA_1}\,\parb{f_3(H)M_{a}f_3(H)}^2.
 \end{equation}
  Assuming that  the  positive parameter
  $\epsilon$ in the construction of the operators $M_{a}$ is
  sufficiently small, it then follows that 
  \begin{equation}
    \label{eq:small}
    \chi_-\parb{\tfrac{8 n}{\epsilon_0^2}M}\chi_+(B/\epsilon_0)f_1(H)=\vO(\inp{x}^{-1/2});\q
    n=\# \vA_1.
  \end{equation}
\end{lemma}
\begin{proof} 
Thanks to \cite[Lemma
5.1]{Sk2}  (which is a consequence of  Lemma \ref{lemma:m1}
\ref{item:15b})  we can record the bound
 \begin{equation}
    \label{eq:small9}
    \chi^2_-(2M_{a_0}/\epsilon_0)\chi_+(B/\epsilon_0)f_1(H)=\vO(\inp{x}^{-1/2}).
  \end{equation}

 Introducing 
  \begin{equation*}
    S=f_2(H)\chi_-\parb{\tfrac{8 n}{\epsilon_0^2}M}
  \chi_+^2(2M_{a_0}/\epsilon_0)\inp{x}^{1/2}f_1(H),
  \end{equation*} it then suffices 
thanks to \eqref{82a0} 
 to  show that
  $S=\vO(\inp{x}^{0})$.

  We estimate for any $\brp\in L^2_\infty\subseteq \vH$ by commutation
  using \eqref{82a0}:
  \begin{align} \label{eq:commsmall0}
  \begin{split}
&\tfrac{\epsilon_0}3 \norm {{S\brp}}^2\\
&\leq \inp{2M_{a_0}-\epsilon_0 I
  }_{S\brp}+C_1\norm{\brp}^2 \\
&\leq -\epsilon_0 \norm{S\brp}^2+2\Sigma_{a\in
  \vA_1}\norm{f_3(H)M_af_3(H)S\brp}\,\norm{S\brp}+C_1\norm{\brp}^2 \\&
\leq -\epsilon_0 \norm{S\brp}^2+\tfrac{2 n}{\epsilon_0}\inp{M}_{S\brp}+\tfrac{\epsilon_0 }2\norm{S\brp}^2
+C_1\norm{\brp}^2 
\\&\leq -\tfrac{\epsilon_0 }2\norm{S\brp}^2+\tfrac{4}8\epsilon_0 \norm{S\brp}^2
+C_2\norm{\brp}^2 =C_2\norm{\brp}^2. 
 \end{split} 
\end{align}

By repeating  the estimation
\eqref{eq:commsmall0}  with $S\brp$ replaced by 
$\inp{x}^{s}S\inp{x}^{-s}\brp$, $s\in \R\setminus\set{0}$,  we
conclude  the desired zero  order 
 estimate.
\end{proof}

 As in Lemma \ref{lem:compBogB} we let $f_1,f_2$ and $f_3$ be standard support functions with
$f_3\succ f_2\succ f_1 $, now assuming  the additional property that $f_1=1$ in a neighbourhood of
$\lambda$.  Let $h\in C^\infty(\R)$ be the function
  $h(t)=\chi_+^2\parb{\tfrac{8 n}{\epsilon_0^2}t}/t$ with  $n=\# \vA_1$. 

 Thanks to \eqref{eq:MicroB} (with
`plus') and Lemma \ref{lem:compBogB}  we  conclude that 
\begin{equation*}
  \phi-f_2(H)h(M){M} A^2_{1}
   f_1(H)\phi\in \vB_0^*;\q \phi=R(\lambda+\i
  0)\psi,\, \psi\in L^2_\infty.
\end{equation*} 

Hence our goal is to extract 
the asymptotics of the second term. Of course we can assume that
$\psi= f_1(H)\psi.$ After commutation it then suffices to consider
the sum 
\begin{equation*}
  \sum_{a\in \vA_1}\,f_2(H)\,h(M) M_{a} A^2_{1}M_{a} R(\lambda+\i0)\psi.
\end{equation*} 

Recalling that $\vA_1=\vA_2\cup
  \set{a_{\min}}$ the above sum splits into the  sum over $\vA_2$
  and the contribution from $a_{\min}$. The latter is given as 
  $\Psi_{a_{\min}}\phi$ with
  \begin{subequations}
\begin{equation}\label{eq:frepHI}
  \Psi_{a_{\min}}=f_2(H)\,h(M) M_{a_{\min}} A^2_{1}M_{a_{\min}}
  f_2(H). 
\end{equation}
 For $a\in \vA_2$ one  easily verifies  by further commutation that
\begin{align*}
  f_2(H)\,h(M)&M_a \parb{A^2_{1}-N^a}
  M_a f_2(H)\phi- S^a_1\phi-S^a_2\phi\in\vB_0^*;\\
S^a_1&=f_2(H)\,h(M)M_a {A_{1}\chi^2_+\parb{
  r^{\rho_1/2}B_\delta^ar^{\rho_1/2}}A_{1}}
  M_a f_2(H),\\
S^a_2&=f_2(H)\,h(M)M_a {A_{1}\chi_+\parb{
    r^{\rho_2-1}r_\delta^a}(A^a_{3})^2\chi_+\parb{
  r^{\rho_2-1}r_\delta^a}A_{1}}
  M_a f_2(H).
\end{align*} 
Introducing similarly the  zero order operators
\begin{equation}\label{eq:stachaModbb}
  \Psi_a=f_2(H)\,h(M) M_a N^a M_a
    f_2(H)\mand \Psi_0= \sum_{a\in\vA_2}\parb{S^a_1+S^a_2}, 
\end{equation} 
 we conclude that 
\begin{equation}\label{eq:phasDEC}
  \phi-\Psi_{a_{\min}}\phi-\sum_{a\in\vA_2} \Psi_a\phi -\Psi_0\phi\in \vB_0^*.
\end{equation} 
\end{subequations}

Effectively $\Psi_{a_{\min}}\phi$ and $\Psi_0 \phi$ are supported in
the `free region' and should therefore have corresponding outgoing
free asymptotics, although the second term is more subtle than the
first one. On the other hand $\Psi_a\phi$, $a\in\vA_2$, should have
asymptotics given by outgoing quasi-modes in the variable $x_a$.  We
will confirm this picture by an  analysis involving  notation and results from
Appendix \ref{sec:Resolvent estimates}. The terms
$\Psi_{a_{\min}}\phi$, $\Psi_a\phi$ and $\Psi_0\phi$  are treated 
in
Subsections \ref{subsec:Free channel term}, \ref{subsec:Near a collision plane} and \ref{subsec:Free
  channel contribution}, 
respectively.

\subsection{Easy free channel term $\Psi_{a_{\min}}\phi$}\label{subsec:Free channel term}
In this subsection we show  that the contribution  to $\phi$ from  the second  term
$\Psi_{a_{\min}}\phi$ in
\eqref{eq:phasDEC}   conforms  with \eqref{eq:asres29}.

Recalling  the
operator $\brH_{a_{\min}}$ and $\brR_{a_{\min}}(z)$ from \eqref{eq:brevH} we  let
\begin{align*}
  \breve\Psi_{a_{\min}}=f_2(\brH_{a_{\min}}) h(M) M_{a_{\min}}  A^2_{1} M_{a_{\min}}
    f_2(H),
\end{align*} and  note that 
\begin{equation*}
  \parb{   \Psi_{a_{\min}}- \breve \Psi_{a_{\min}}}\phi\in \vB_0^*.
\end{equation*}  
  Thanks to the properties 
\begin{equation}\label{eq:thre00}
   \vT(\brH_{a_{\min}})=\set{0}\mand \sigma_{\pupo}(\brH_{a_{\min}})= \sigma_{\pupo}( \brh_{a_{\min}})\subseteq (-\infty, 0],
\end{equation} 
 there is a Mourre estimate for $\brH_{a_{\min}}$
at the  positive energy $\lambda$.  
  By a  resolvent equation  we are consequently lead to   
write  (here computing formally)
\begin{align}\label{eq:res100}
  \begin{split}
  \breve \Psi_{a_{\min}}\phi=  \breve \phi_{a_{\min}}:=&\brR_{a_{\min}}(\lambda+\i
  0)
  \breve\psi_{a_{\min}};\\\quad\breve\psi_{a_{\min}}&=\breve\Psi_{a_{\min}}\psi-\i
  \brT_{a_{\min}}\phi,\quad \brT_{a_{\min}}=\i  \parb{\brH_{a_{\min}}\breve\Psi_{a_{\min}}-\breve\Psi_{a_{\min}} H},\\
\quad\breve\Psi_{a_{\min}}\psi  &\in L^2_{\infty},\quad
  \brT_{a_{\min}}=   \vO(\inp{x}^{-1}).  
  \end{split}
\end{align} 

The  complete justification of \eqref{eq:res100}
depends on  Appendix
\ref{sec:Resolvent estimates} and will not be given, rather  we will
elaborate on  the main ingredients  only, to be done below. 
A  very similar (although more complicated) problem for $\Psi_0\phi$ is treated in
detail in 
Subsection \ref{subsec:Free channel contribution} with proper
reference to  Appendix
\ref{sec:Resolvent estimates}. Hence let us here  just  note that a possible
(and correct) interpretation of
\eqref{eq:res100} is given as
\begin{equation*}
  \breve\psi_{a_{\min}}\in L^2_{1/2}\mand \breve \phi_a=\brR_{a_{\min}}(\lambda+\i
  0)\breve\psi_{a_{\min}}=\lim_{\epsilon\to 0_+}\,\brR_{a_{\min}}(\lambda+\i
  \epsilon)\breve\psi_{a_{\min}}\text{ weakly in  }L^2_{-1}.
\end{equation*} 
Moreover we take for granted the existence of   a
sequence $L^2_\infty\ni \breve\psi_{a_{\min},n}\to \breve\psi_{a_{\min}}\in \vH$ with
convergence
\begin{equation}\label{eq:com200}
  \brR_{a_{\min}}(\lambda+\i
  0) \breve\psi_{a_{\min},n}\to  \brR_{a_{\min}}(\lambda+\i 0)\breve\psi_{a_{\min}}\text{ in }\vB^*\text{ for }n\to \infty.
\end{equation}

The first point to  record  is then that for each $n\in \N$ there exists  $g_n\in
\vG_{a_{\min}}=L^2({\S^{d-1}})$ such that 
\begin{equation}\label{eq:asres2oneOnb00}
\brR_{a_{\min}}(\lambda+\i
  0) \breve\psi_{a_{\min},n}
-2\pi \i \, \breve v^{+}_{\alpha_{\min},\lambda} [g_n]\in \vB^*_0,
\end{equation} see for example \cite{IS4}.
Here the second term  is 
labelled by the `free channel' $\alpha_{\min}$
defined uniquely for $a=a_{\min}$. (The  asymptotics
\eqref{eq:asres2oneOnb00} is more complicated to derive in the
context of Subsection \ref{subsec:Free channel contribution} since the
 classical conditions on the one-body potential are not
available there.)

By combining  \eqref{eq:com200}, 
  \eqref{eq:asres2oneOnb00} and the elementary computation
\begin{equation}\label{eq:eqBnd00}
  \norm[\big] { 
    \breve v^{+}_{\alpha_{\min},\lambda}
    [g]}^2_{\vB^*/\vB_0^*}=(4\pi)^{-1} \lambda^{-1/2}\norm{g}^2,
\end{equation}
 we conclude that there exists  $g\in \vG_{a_{\min}}$ such that
 $g_n\to g\in \vG_{a_{\min}}$,  which in turn   yields  (by taking
 $n\to \infty$) that 
\begin{equation}\label{eq:asres2oneOn00}
       \brR_{a_{\min}}(\lambda+\i
  0)\breve\psi_{a_{\min}}-2\pi \i \, \breve v^{+}_{\alpha_{\min},\lambda} [g]\in \vB^*_0
 \text{ for some }g\in \vG_{a_{\min}}.
    \end{equation} Consequently     the contribution to $\phi$ from
    $\Psi_{a_{\min}}\phi$   conforms  with
 \eqref{eq:asres29}, as wanted.

 The operator $\brT_{a_{\min}}$ in \eqref{eq:res100} has order
 $ \vO(\inp{x}^{-1})$, which just misses application of the limiting
 absorption principle bound \eqref{eq:LAPbnda} (for $H$ as well as for
 $\brH_{a_{\min}}$). Hence a more  detailed   computation of the
 operator is needed. Let us sketch it. First we note that 
\begin{align*}
  \breve\Psi_{a_{\min}}=f_2(\brH_{a_{\min}}) h(M) \tilde \xi^+_{a_{\min}}M_{a_{\min}}  A^2_{1} M_{a_{\min}}
    f_2(H),
\end{align*} where $\tilde\xi^+_{a_{\min}}$ is given as in 
\eqref{eq:partM} (with $\varepsilon>0$ chosen small as required for
\eqref{eq:partM}). From the construction \eqref{eq:Ipott} it follows
that 
the function 
 $\parb{\breve I_{a_{\min}}- I^{\rm lr}_{a_{\min}}}\tilde\xi^+_{a_{\min}}$ has compact support.
 After commutation this allows us to replace the factor of
 $f_2(\brH_{a_{\min}}) $ by $f_2(H)$. In fact the order of  the
 resulting 
 difference is  $ \vO(\inp{x}^{-\infty})$, since once a commutation introduces
 a derivative of $\tilde\xi^{+}_{a_{\min}}$, say denoted by
 $\tilde\xi_{a_{\min}}'$,  we can write
 $M_{a_{\min}}=\xi^+_{a_{\min}}M_{a_{\min}} $, commute and use that
 $\tilde\xi_{a_{\min}}'\xi^+_{a_{\min}}=0$. In conclusion, thanks to   the
 presence of a factor $M_{a_{\min}}$  the factors
 $f_2(\brH_{a_{\min}})$ and $f_2(H)$ can freely be interchanged (this
 will more generally be
 used in both directions), and hence we are led 
 to consider the commutator
 \begin{align*}
   \i  \parb{H\Psi_{a_{\min}}-\Psi_{a_{\min}} H}&=f_2(H)
   \ad_{\i g(H)}\parb{ h(M)} M_{a_{\min}}  A^2_{1} M_{a_{\min}}
    f_2(H)\\
&+f_2(H)h(M)
   \ad_{\i g(H)}\parb{ M_{a_{\min}}}  A^2_{1} M_{a_{\min}}
    f_2(H)\\
&+f_2(H)h(M)M_{a_{\min}}
   \ad_{\i g(H)}\parb{  A^2_{1}}  M_{a_{\min}}
    f_2(H)\\
&+f_2(H)h(M)M_{a_{\min}} A^2_{1}
   \ad_{\i g(H)}\parb{  M_{a_{\min}}} 
    f_2(H).
 \end{align*} Here $g(\lambda) =\lambda f_3(\lambda)$ (recall that
 we have fixed standard support functions
$f_3\succ f_2\succ f_1 $). 

With a proper application of \eqref{82a0} the first, second  and the
fourth terms are treated by the limiting
 absorption principle bound \eqref{eq:LAPbnda} and by
 \eqref{eq:Hes0}--\eqref{eq:2boundobtain33}   (valid for $H$) as well as
 (\ref{eq:2boundobtain3350}) and \eqref{eq:2boundobtain33500} (valid for 
 $\brH_{a_{\min}}$). For example, we can for the first term write $ M=\Sigma_{b\in
   \vA_1}\,\parb{f_3(H)M_{b}f_3(H)}^2$, and using \eqref{82a0} the commutator
$\i[H,M_b]$ will appear as a factor for each term in an expansion. Therefore in turn a factor
$Q(b,j)^*\vG_jQ(b,j)$ as in \eqref{eq:MaLeading1} will appear. The
factor $Q(b,j)$ to the right is then treated by
\eqref{eq:2boundobtain33}, while the factor $Q(b,j)^*$ to the left  is
treated by (\ref{eq:2boundobtain3350}) and
\eqref{eq:2boundobtain33500}. (The use of \eqref{eq:2boundobtain33500}
in our  paper is actually limited to treating $\ad_{\i g(H)}\parb{
  h(M)}$ this way.) 

The third term is treated (after first applying
 \eqref{82a0}) by the  $Q$-bound 
\begin{equation}\label{eq:2boundobtain3399}
  \sup _{\Im z\neq 0}\norm{Q{f_2}
   (H){R(z)}}_{\vL(\vB,\vH)}< \infty,\q Q=r^{-1/2}\sqrt{(\chi^2_+)'}\parbb{ B/\epsilon_0}
\end{equation} and its analogue with $H$ replaced by
$\brH_{a_{\min}}$, cf. \cite[(2.11b)]{Sk2}.

We conclude from the above computations that for bounded  (computable)
operators $\breve B_1, \breve B_2, \dots $ and for explicit `$Q$-operators'
\begin{equation}
  \label{eq:form00}
  \breve\psi_{a_{\min}}=
\sum_{Q_k} \,f_2 (\brH_{a_{\min}}){\breve Q_k}^*\breve B_k
 Q_k f_2 (H) R(\lambda+\i 0)\psi+\psi';
\quad \psi'\in\vB.
\end{equation} Upon substituting into \eqref{eq:res100} the obtained
representation of $\breve \Psi_{a_{\min}}\phi$ is  a
mathematically valid representation (to be  demonstrated for an 
analogous  model in Subsection \ref{subsec:Free channel
  contribution}). Note that thanks to the mentioned $Q$-bounds
indeed the
terms  
\begin{equation*}
  \parbb{\brR_{a_{\min}}(\lambda+\i 0)f_2 (\brH_{a_{\min}}){\breve Q_k}^*}\breve B_k
 \parbb{Q_k f_2 (H) R(\lambda+\i 0)\psi}
\end{equation*} are   well-defined elements of $\vB^*$. The
above approximation property \eqref{eq:com200}  follows easily from
this representation, see \eqref{eq:form} and \eqref{eq:n2} for an
elaboration in a similar context.

\subsection{Collision plane terms  $\Psi_a\phi$, $a\in \vA_2$}\label{subsec:Near a collision plane}
We show  that the contribution  to $\phi$ from any of the terms
$\Psi_a\phi$ in
\eqref{eq:phasDEC} with  $a\in \vA_2$  conforms  with \eqref{eq:asres29}. 

Recalling  the
operator $\brH_{a}$ and $\brR_a(z)$ from \eqref{eq:brevH} we  let for any 
fixed  $a\in \vA_2$ 
\begin{align*}
  \breve\Psi_a=f_2(\brH_a) h(M) M_a N^a M_a
    f_2(H),
\end{align*} and  note that   $\abs{x_a}\to \infty$ on $\supp
(A^a_2)$ when  $\abs{x}\to \infty$, and consequently that 
\begin{equation*}
  \parb{   \Psi_a- \breve \Psi_a}\phi\in \vB_0^*.
\end{equation*}  
  Thanks to  \eqref{eq:3body} and \eqref{eq:sim2}  the threshold set of $\brH_a$ is given as
  \begin{subequations}
    \begin{equation}\label{eq:ppH}
  \vT(\brH_a)= \sigma_{\pupo}( \brh_a)\cup \sigma_{\pupo}( H^a)\cup\set{0}\subseteq (-\infty, 0].
\end{equation}  Similarly 
\begin{equation}\label{eq:thre}
   \sigma_{\pupo}(\brH_a)= {\set{E=\lambda_1+\lambda_2|\,
       \lambda_1\in \sigma_{\pupo}( H^a),\, \lambda_2\in
       \sigma_{\pupo}( \brh_a)}}\subseteq (-\infty, 0].
\end{equation} 
  \end{subequations}
 Hence  there is a Mourre estimate for $\brH_a$
at the  positive energy $\lambda$.  
  By a resolvent equation  we are consequently led to   
write  (here computing formally)
\begin{align}\label{eq:res1}
  \begin{split}
  \breve \Psi_a\phi&=  \breve \phi_a:=\brR_a(\lambda+\i
  0) \breve\psi_a;\\&\quad\breve\psi_a={\breve\Psi_a\psi  -\i \brT_a\phi},\quad
  \brT_a=\i  \parb{\brH_a\breve\Psi_a-\breve\Psi_a H}=    \vO(\inp{x}^{\rho_1-\delta}).  
  \end{split}
\end{align} 

{\bf I} (justifying  \eqref{eq:res1}).  Although the indicated
order of $\brT_a$ appears too weak for applying \eqref{eq:BB^*a} the formula
\eqref{eq:res1} turns out to
be correct by Mourre estimates and their consequences \cite{{AIIS}}
and  a
variety of weak type estimates of Appendix \ref{sec:Resolvent
  estimates} (similar to those of \cite {Sk1,Sk2}).

We will prove 
that the function $\breve \phi_a$ in \eqref{eq:res1} is 
well-defined  as an element in $\vB^*$. This involves  an   a priori interpretation 
different from \eqref{eq:BB^*a}. From the outset $\breve \phi_a$ is the  weak
limit,   say in
 $L^2_{-1}$,
\begin{align*}
  \breve \phi_a= \lim_{\epsilon\to 0_+}  \brR_a(\lambda+\i
 \epsilon )  &\breve\Psi_a\psi    -\i \lim_{\epsilon\to 0_+}\brR_a(\lambda+\i
 \epsilon )\brT_aR(\lambda+\i \epsilon)\psi \\=\brR_a(\lambda+\i
 0 )  \breve\Psi_a\psi    -\i \lim_{\epsilon\to 0_+}&\brR_a(\lambda+\i
 \epsilon )\chi^2_-(2B/\epsilon_0)\brT_aR(\lambda+\i \epsilon)\psi\\&-\i \lim_{\epsilon\to 0_+}\brR_a(\lambda+\i
 \epsilon )\chi^2_+(2B/\epsilon_0)\brT_aR(\lambda+\i \epsilon)\psi. 
\end{align*}  
 It follows from \eqref{eq:BB^*a} that the first term to the right is
 an element in $\vB^*$ (since $\breve\Psi_a\psi\in \vB$). By
 commutation we can write (for the second  term)
 \begin{equation*}
   \chi^2_-(2B/\epsilon_0)\brT_a=\inp{x}^{-1}\breve B\inp{x}^{-1}\text{ with
   }\breve B \text{ bounded}.
 \end{equation*} This allows us to compute
 \begin{align*}
   \lim_{\epsilon\to 0_+}&\brR_a(\lambda+\i
 \epsilon )\chi^2_-(2B/\epsilon_0)\brT_aR(\lambda+\i \epsilon)\psi\\&=\lim_{\epsilon\to 0_+}\brR_a(\lambda+\i
 \epsilon)\chi^2_-(2B/\epsilon_0)\brT_aR(\lambda+\i
                                                                      0)\psi\quad(\text{by
                                                                      }\eqref{eq:LAPbnda}\\
&=\brR_a(\lambda+\i 0)\chi^2_-(2B/\epsilon_0)\brT_aR(\lambda+\i
                                                                      0)\psi\in
  \vB^*\quad(\text{by }\eqref{eq:BB^*a}.
 \end{align*} We conclude that the first and second terms
 are   on  a  form consistent with \eqref{eq:BB^*a},
\begin{equation*}
      \brR_a(\lambda+\i
 0)\psi'\in \vB^*
      \text{ for some }\psi'\in (L^2_1\subseteq)\, \vB.
    \end{equation*}
 
The third term is different. 
The  expression
$\chi^2_+(2B/\epsilon_0)\brT_aR(\lambda+\i 0)\psi$ defines  an element
of 
$L_s^2$, $2s=\delta-\rho_1\in (1/3,1)$, see \eqref{eq:form} and
\eqref{eq:orderQ} below.  However  we do not  prove  better decay.

Thanks to 
\eqref{eq:LAPbnda} and \cite[Theorem  1.8] {{AIIS}} there exists the operator-norm-limit
\begin{equation}\label{eq:strongED_+SDY}\vLlim_{\epsilon\to 0_+}\,\,
  \inp{x}^{-1}\brR_a(\lambda+\i
 \epsilon)\chi^2_+(2B/\epsilon_0)\inp{x}^{-s},\quad s=(\delta-\rho_1)/2,
\end{equation} cf. \cite[Appendix C]{Sk1}. 

Next we invoke  Appendix \ref{sec:Resolvent
  estimates}. By  a  `pedestrian' (although  lengthy)   expansion 
 the operator $\brT_a$ is seen to be  a sum of terms on 
 the form  $f_2 (\brH_a){\breve Q_k}^*\breve B_k
 Q_k f_2 (H)$, as  specified in Appendix \ref{sec:Resolvent estimates} (see Subsection \ref{subsec:Free channel term}
 for a treatment of  a simpler case). In agreement with
 \eqref{eq:Qsquared} the index  $k=1,\dots,10$ labels the different occurring
 forms of  $Q$-operators and in all cases $\breve B_k$ is bounded.
Using \eqref{eq:kato10},  \eqref{eq:orderQ}  and
\eqref{eq:LAPbnda}  we can take  the weak limit
\begin{equation}
  \label{eq:weak} Q_k f_2 (H)R(\lambda+\i 0)\psi:=\wvHlim_{\epsilon\to 0_+}\,\, Q_k f_2 (H)R(\lambda+\i \epsilon)\psi.
\end{equation}
Using that $\inp{x}^{s}f_2 (\brH_a){\breve Q_k}^*\breve B_k$
is bounded, \eqref{eq:strongED_+SDY} and (\ref{eq:weak}), we can compute
the above third term
 as follows. Taking limits in the weak sense in $L^2_{-1}$
\begin{align*}
  -\i &\lim_{\epsilon\to 0_+}\brR_a(\lambda+\i
 \epsilon )\chi^2_+(2B/\epsilon_0)\brT_aR(\lambda+\i
        \epsilon)\psi\\
&=-\i \lim_{\epsilon\to 0_+}\brR_a(\lambda+\i
  0)\chi^2_+(2B/\epsilon_0)\brT_aR(\lambda+\i \epsilon)\psi\quad\text{
  by \eqref{eq:strongED_+SDY}} \\
&=-\i\brR_a(\lambda+\i 0 )\chi^2_+(2B/\epsilon_0)\brT_aR(\lambda+\i
  0)\psi\quad\text{ by \eqref{eq:weak} }\\
&=-\i \lim_{\epsilon\to 0_+}\brR_a(\lambda+\i
 \epsilon )\chi^2_+(2B/\epsilon_0)\brT_aR(\lambda+\i
        0)\psi\quad\text{ by \eqref{eq:strongED_+SDY} }.
\end{align*} 

We may summerize our interpretation  of \eqref{eq:res1} as
\begin{equation*}
  \breve\psi_a\in (L^2_{s}\subseteq) \,\vH \mand \breve \phi_a=\brR_a(\lambda+\i
  0)\breve\psi_a=\lim_{\epsilon\to 0_+}\,\brR_a(\lambda+\i
  \epsilon)\breve\psi_a\text{ weakly in  }L^2_{-1}.
\end{equation*} 

Next we improve on this assertion  by claiming the existence of a
sequence $L^2_\infty\ni \breve\psi_{a,n}\to \breve\psi_a\in \vH$ with
convergence
\begin{equation}\label{eq:com2}
  \brR_a(\lambda+\i
  0) \breve\psi_{a,n}\to  \brR_a(\lambda+\i
  0)\breve\psi_a\text{ in }\vB^*\text{ for }n\to \infty.
\end{equation} 

To construct such regularization we first note the following form of
the vector $\breve\psi_a$, which follows  from the above
discussion and the explicit form of the operators $\breve Q_k$ in
Appendix \ref{sec:Resolvent estimates}. We decompose into a finite sum
(where  for each term 
the involved $Q$-operator is on one of the ten forms listed in
Appendix \ref{sec:Resolvent estimates})
\begin{subequations}
\begin{equation}
  \label{eq:form}
  \breve\psi_a=
\sum_{Q_k } \,f_2 (\brH_a){\breve Q_k}^*\breve B_k
 Q_k f_2 (H) R(\lambda+\i
        0)\psi+\psi';
\quad \psi'\in\vB.
\end{equation} Introducing $\chi_n=\chi_-(r/n)$  we are led to define,
cf.
 the proof of \cite [Lemma  9.12]{Sk1},
\begin{equation}\label{eq:n2}
  \breve \psi_{a,n}=
  \sum_{Q_k } \,f_2 (\brH_a){\breve Q_k}^*\chi_n\breve B_k
  Q_k f_2 (H) R(\lambda+\i
  0)\psi+\chi_n\psi'.
\end{equation} 
\end{subequations}
 Due  to \eqref{eq:BB^*a}, (\ref{eq:weak}),
\eqref{eq:phiweak}
 and  the facts  that
$\chi_n\breve B_k \to \breve B_k $ strongly on $\vH$ and
$\chi_n\psi'\to \psi'$ in $\vB$,
\begin{equation*}
  \sup_{\epsilon>0}\,\norm [\big]{\brR_a(\lambda+\i
  \epsilon) \parb{\breve\psi_{a}-\breve\psi_{a,n}}}_{\vB^*}\to 0\text{ for }n\to \infty.
\end{equation*}
 This  uniform convergence  yields  \eqref{eq:com2}.

{\bf II} (smoothness- and $\vB^*$-estimates).   
We need a variation of the bounds 
\cite [(7.3a) and (7.3b)]{Sk1}. By using the same  `propagation
observables' as for \eqref{eq:phiweak} we obtain by mimicking the proof
of  \cite [Lemma 7.1]{Sk1} that 
\begin{equation*}\label{eq:kato2}
  \breve Q_l{f_2} (\brH_a)\delta(\brH_a-\lambda){f_2} (\brH_a){\breve Q_k}^* \in
   \vL(\vH).
\end{equation*} 
 
 In combination with \eqref{eq:BB^*a} this leads to the following assertion for the  approximating
sequence of  \eqref{eq:n2}:
 \begin{subequations}
 \begin{equation}\label{eq:brevNapprox}
   \inp{\delta(\brH_a-\lambda)}_{\breve\psi_a-\breve\psi_{a,n}}\to
   0\text{ for }n\to \infty.
 \end{equation}

We will use these features below in the  computation  of 
 \begin{align*}
   \phi_a:=\sum_{\alpha=(a,\lambda^\alpha,
     u^\alpha),\,(H^a-\lambda^\alpha)u^\alpha=0,\,\lambda^\alpha<\lambda}
   2\pi\i  \,J_\alpha&
    \breve v^{+}_{\alpha,\lambda} [g_\alpha]-\brR_a(\lambda+\i
  0)\breve\psi_a \in \vB^*/\vB_0^*;\\&\quad
    g_\alpha:=\breve\Gamma_\alpha^+(\lambda)\breve\psi_a\in \vG_a.
 \end{align*} The operator 
 $\breve\Gamma_\alpha^+(\lambda)$ is the
 outgoing 
 restricted $\alpha$-channel wave operator for $\brH_a$ at energy $\lambda$, possibly
 defined as in Proposition \ref{prop:radi-limits-chann22} (with $H$
 replaced by $\brH_a$).  Here and henceforth we only consider channels $\alpha$ specified
 as in the summation (in particular with the
 first 
 component fixed as $a$). We argue that $\phi_a$ is   a well-defined:
 Thanks to
 the above discussion and arguments from \cite [Subsection 9.2]{Sk1}
 it follows that indeed $g_\alpha$ is a well-defined element of $\vG_a$. In
 fact  by   the resulting  extension of the Bessel inequality
 \eqref{eq:Besn}  for $\brH_a$, 
\begin{equation}\label{eq:ScatEnergy2233}
  \sum_{\lambda^\alpha<
  \lambda}\,\norm{g_\alpha}^2 \leq 
  \inp{\delta(\brH_a-\lambda)}_{\breve\psi_a}<\infty,
\end{equation} and 
 the general bounds, cf. \cite [(9.22)]{Sk1},
\begin{equation}\label{eq:eqBnd}
  C_1\sum_{\lambda^\alpha<
    \lambda}\,\norm{g_\alpha}^2 \leq\norm[\Big] {\sum_{\lambda^\alpha<\lambda} J_\alpha
    \breve v^{+}_{\alpha,\lambda} [g_\alpha]}^2_{\vB^*/\vB_0^*}\leq C_2\sum_{\lambda^\alpha<
    \lambda}\,\norm{g_\alpha}^2,
\end{equation} it follows that   $\phi_a$ is   a well-defined 
element of $\vB^*/\vB_0^*$.

 Next, we introduce  $\phi_{a,n}$ and $g_{\alpha,n}$ by replacing  in the above
expressions  for $\phi_{a}$ and $g_{\alpha}$
 all appearances of $\breve\psi_a$  by $\breve\psi_{a,n}$. 
 Similarly to \eqref{eq:ScatEnergy2233},  we can then record that 
\begin{equation}\label{eq:ScatEnergy2233b}
  \sum_{\lambda^\alpha<
  \lambda}\,\norm{g_\alpha-g_{\alpha,n}}^2 \leq 
  \inp{\delta(\brH_a-\lambda)}_{\breve\psi_a-\breve\psi_{a,n}}.
\end{equation} 
  \end{subequations}

{\bf III} (applying the  estimates).   Writing for given $\phi_1,\phi_2\in\vB^*$, 
$\phi_1 \simeq \phi_2$ if $\phi_1 - \phi_2 \in\vB^*_0$, we compute
considering now $\phi_a\in\vB^*$ as a representative for the corresponding
coset (see below for further elaboration)
\begin{align}\label{eq:long}\begin{split}
  &\phi_a\\
&\simeq \chi_-\parb{
  r^{\rho_2-1}r_\delta^a/2}\phi_a\\
&\simeq \chi_-\parb{
  r^{\rho_2-1}r_\delta^a/2}
\phi_{a,n}+o(n^0)\quad\quad\parb{\text{replacing 
  }\breve\psi_a\text{ by }\breve\psi_{a,n}}\\
& \simeq\chi_-\parb{
  r^{\rho_2-1}r_\delta^a/2}
\chi_-(mH^a)\phi_{a,n}+o(n^0)\quad\quad(\text{by  velocity bounds})\\
&\simeq o(n^0) +o(m^0)\quad\quad(\text{by  dominated convergence and spectral
  theory})\\
&\simeq 0.
\end{split}
\end{align} In the first step of (\ref{eq:long}) we used the appearance
of factors of $A^a_2$ in the definition of $\breve \Psi_a$, in the
second step \eqref{eq:com2} and
\eqref{eq:brevNapprox}--\eqref{eq:ScatEnergy2233b}, in the third step
a stationary  version of the so-called  minimal velocity bound 
(recalled in Appendix \ref{sec:Non-threshold analysis})
 and in the last steps we first fixed
a big $n$ and then $m=m(n)$ sufficiently big to conclude that the
distance from $\phi_a$ to $\vB^*_0$ is at most $\epsilon$, for any
prescribed $\epsilon>0$.  However we need to argue that for fixed
(big) $n$ indeed the approximation by taking $m$ correpondingly big
works.

 For this purpose we    record that with  $f^a_{m}:= 1_{\vT_\p(H^a)}-\chi_-(m\cdot)$
  \begin{subequations}
  \begin{align}\label{eq:2pp} 
    \begin{split}
   \norm {f^a_{m}&(H^a)\brR_a(\lambda+\i
     0) \breve\psi_{a,n}}_{\vB^*}\leq C_m\norm
                         {f^a_{m}(H^a)\inp{x_a}\breve\psi_{a,n}}_{\vH};\\
& C_m=\sup_{\inf \vT_\p(H^a)\leq \eta\leq 2/m}\, \norm{\brr_a(\lambda-\eta+\i
 0) \inp{x_a}^{-1}}_{\vL(L^2(\bX_a),\,\vB(\bX_a)^*)}.   
    \end{split}
  \end{align} Note that  (\ref{eq:2pp}) follows  by  using  the trivial inclusion
  $\set{\abs{x}\leq \rho}\subseteq \set{\abs{x_a}\leq \rho}$ and the
  spectral theorem on multiplication operator form for $H^a$
  \cite[Theorem VIII.4]{RS} (amounting to a partial
  diagonalization).  The
  sequence $(C_m)_1^\infty$ is  bounded (since it is decreasing). Rewriting $\vH=L^2\parb{\bX_a,L^2(\bX^a)}$
  and invoking the dominated convergence theorem and the Borel
  calculus for $H^a$, it follows  that
\begin{equation}\label{eq:firDec}
   \norm{f^a_{m}(H^a)\brR_a(\lambda+\i
     0) \breve\psi_{a,n}}_{\vB^*}\to 0\text{ for }m\to \infty.
\end{equation}

  Next we claim that 
\begin{equation}\label{eq:1pp}
    2\pi\i  \sum_{\lambda^\alpha<\lambda} J_\alpha
    \breve v^{+}_{\alpha,\lambda} [\breve\Gamma_\alpha^+(\lambda)\breve\psi_{a,n}]\simeq 1_{\vT_\p(H^a)}(H^a)\brR_a(\lambda+\i
 0) \breve\psi_{a,n}.
  \end{equation} 
\end{subequations} Note that we consider the left-hand side as an
element of $\vB^*$ although its strict meaning is the correponding
coset in  $\vB^*/\vB_0^*$. Thanks to  \eqref{eq:2pp},
also the right-hand side is well-defined in $\vB^*$ (or in $\vB^*/\vB_0^*$). 

To show  \eqref{eq:1pp} we first consider the simplest case where  there are only finitely many channels
  $\alpha$ involved in the summation defining $\phi_a$.  We  use that
  $\breve\Gamma_\alpha^+(\lambda)=
  \brg_\alpha^+(\lambda_\alpha)J^*_\alpha$, as expressed by the one-body
  restriction operator $ \brg_\alpha^+(\cdot )$  for $\brh_a$  at the   positive energy 
  $\lambda_\alpha=\lambda-\lambda^\alpha$ (cf.  Theorem \ref{thm:wave_matrices} and
  \eqref{eq:wave_op30}),  and  stationary completeness
   of positive energies 
  for    one-body Schr\"odinger
  operators, defined similarly and definitely valid for  $\brh_a$
   (cf.  \cite {Sk1}). This means concretely  that
  \begin{equation*}
    (\brh_a-\lambda_\alpha -\i 0)^{-1} \brf-
    2\pi\i \,\breve v^{+}_{\alpha,\lambda}
    [\brg_\alpha^+(\lambda_\alpha)\brf]\in \vB^*_0(\bX_a);\quad \brf=J^*_\alpha\breve\psi_{a,n}.
  \end{equation*} We then conclude \eqref{eq:1pp} by expanding the
  right-hand side into  a (finite) sum.

In the remaining  case where there are  infinitely many channels
  $\alpha$ in the summation  we pick an increasing
 sequence  of finite-rank projections $P^a_j\to 1_{\vT_\p(H^a)}(H^a)$ (strongly) corresponding
 to any numbering of the channel  eigenstates. By using a
 modification of  \eqref{eq:2pp} (and arguing similarly) we then
 obtain as in \eqref{eq:firDec}  that
 \begin{equation*}
   \norm{\parb{P^a_j-1_{\vT_\p(H^a)}(H^a)}\brR_a(\lambda+\i
 0) \breve\psi_{a,n}}_{\vB^*}\to 0\text{ for }j\to \infty.
\end{equation*} Thanks  to 
this property, a version of (\ref{eq:eqBnd})  and   the
arguments for  the finite summation  case also  the  infinite
summation  formula  follows.
We have proved \eqref{eq:1pp}.

Clearly the combination 
  of \eqref{eq:firDec} and \eqref{eq:1pp} yields the wanted property 
  $\phi_a \simeq 0$. Hence  the contribution to $\phi$ from
  $\Psi_a\phi$, $a\in \vA_2$,   conforms  with
 \eqref{eq:asres29}.

\subsection{Difficult free channel term  $\Psi_0\phi$}\label{subsec:Free channel
  contribution} We   show  that the contribution to $\phi$
from the term $\Psi_0\phi$ in
\eqref{eq:phasDEC} conforms  with
 \eqref{eq:asres29}.

Motivated by the form of $\Psi_0$ and the construction $\breve I_a$ of
Subsection \ref{subsec: -body effective potential and a $1$-body
  radial limit} we
introduce with $\delta>0$ given as in \eqref{eq:parameters} and for a
sufficiently small $c>0$  (given by a property of the operators
$B^a_\delta$, see \cite[Section 5 and (8.15)]{Sk1})
\begin{subequations}
\begin{equation}\label{eq:brevePotential223}
  \widecheck I_0(x)=\chi_+(|x|/R)I^{\rm lr}_{a_{\min}}(x)\prod_{b\in\vA_2} \,\,\chi_+(2\abs{
  x^b}/{cr^\delta}),\quad R\geq 1,
\end{equation} and correspondingly (in this subsection  considering
only $R=1$)
\begin{equation*}
  \widecheck H_0=-\Delta +\widecheck I_0,\mand \widecheck  R_0(z)=(\widecheck H_0-z)^{-1}\text{ for } z\in\C\setminus \R.
\end{equation*} Note that $\widecheck I_0$ is a one-body potential obeying the bounds
\begin{equation}\label{eq:weakerCond}
        \partial^\gamma \widecheck
        I_0(x)=\vO(\abs{x}^{-\delta(\mu+|\gamma|)});\quad \abs{\gamma}\leq2.
      \end{equation}   
\end{subequations}
These are weaker than \eqref{eq:brevePotentialOne}
      for $a=a_{\min}$,  in which case we  abbreviate $\breve I_0$=$\breve I_{a_{\min}}$, however  
      $\widecheck I_0$ is a   classical
 $C^2$ long-range 
 potential  in the terminology of \cite {IS4} thanks to  the conditions \eqref{eq:parameters} (see also the related
       \cite[(2.7.1) and Theorem 2.7.1] {DG}).

       Parallel to Subsection \ref{subsec:A phase-space partition of
         unity}  we introduce
\begin{align*}
\widecheck \Psi_0&=\sum_{a\in\vA_2}\parb{\widecheck S^a_1+\widecheck S^a_2},\\
\widecheck S^a_1&=f_2(\widecheck H_0)h(M)M_a {A_{1}\chi^2_+\parb{
  r^{\rho_1/2}B_\delta^ar^{\rho_1/2}}A_{1}}
  M_a f_2(H),\\
\widecheck S^a_2&=f_2(\widecheck H_0)h(M)M_a {A_{1}\chi_+\parb{
  r^{\rho_2-1}r_\delta^a}(A^a_{3})^2\chi_+\parb{
  r^{\rho_2-1}r_\delta^a}A_{1}}
  M_a f_2(H),
\end{align*} 
 and record  the analogous property
\begin{equation*}
  \parb{   \Psi_0- \widecheck \Psi_0}\phi\in \vB_0^*.
\end{equation*} This follows from the presence of the factors of $\chi_+\parb{
  r^{\rho_1/2}B_\delta^ar^{\rho_1/2}}$ and $\chi_+\parb{
  r^{\rho_2-1}r_\delta^a}$.  More explicitly we use     commutation
and the facts that 
\begin{equation*}
  \chi_+\parb{
  r^{\rho_1/2}B_\delta^ar^{\rho_1/2}}\parb{1-\chi_+(2\abs{
  x^a}/{cr^\delta})}=0,
\end{equation*} cf. \cite [(8.15)]{Sk1}, and that 
\begin{equation*}
  \chi_+\parb{
  r^{\rho_2-1}r_\delta^a} \parb{1-\chi_+(2\abs{
  x^a}/{cr^\delta})}\text{ is compactly supported}.
\end{equation*}

We mimic \eqref{eq:res1} writing 
\begin{align}\label{eq:res1233}
  \begin{split}
    \widecheck \Psi_0\phi&= \widecheck  R_0(\lambda+\i
 0)\widecheck\psi_0;\\&\quad\widecheck\psi_0={\widecheck\Psi_0\psi  -\i \widecheck T_{0}\phi},\quad
\widecheck T_{0}=\i \parb{\widecheck
  H_0\widecheck\Psi_0-\widecheck\Psi_0H}=
\vO(\inp{x}^{\rho_1-\delta}).
\end{split}
\end{align} 
 Parallel to    Subsection \ref{subsec:Near a collision
   plane} we may interprete \eqref{eq:res1233} as 
\begin{equation*}
  \widecheck\psi_0\in \vH \mand \widecheck  R_0(\lambda+\i
 0)\widecheck\psi_0=\lim_{\epsilon\to 0_+}\,\widecheck  R_0(\lambda+\i
 \epsilon)\widecheck\psi_0\text{ weakly in  }L^2_{-1}.
\end{equation*} The proof is  similar to the justification of
 \eqref{eq:res1}.  Note
 that the  bound of \cite[Theorem  1.8] {{AIIS}} is also valid
under \eqref{eq:weakerCond} (may be seen by  computing the second
commutator in the proof of \cite[Theorem  1.8] {{AIIS}} using
\eqref{eq:weakerCond} with  $\abs{\gamma}=2$, rather than
using 
the `undoing trick' of \cite{{AIIS}}). See   Appendix
\ref{sec:Resolvent estimates} for some additional  details.

 Moreover  the following    version
of \eqref{eq:com2} is valid.  There  exists a
sequence $L^2_\infty\ni \widecheck\psi_{0,n}\to \widecheck\psi_0\in \vH$ with
convergence
\begin{equation}\label{eq:com2b}
  \widecheck  R_0(\lambda+\i 0)\widecheck\psi_{0,n} \to \widecheck  R_0(\lambda+\i
 0)\widecheck\psi_0\text{ in }\vB^*\text{ for }n\to \infty.
\end{equation} 
  The proof is  similar to the one of \eqref{eq:com2}. Hence writing (as in 
\eqref{eq:form})
\begin{subequations}
\begin{equation}\label{eq:90}
  \widecheck\psi_{0}=
  \sum_{ Q_k } \,f_2 (\widecheck H_0){\widecheck Q_k}^*\widecheck B_k
  Q_k f_2 (H) R(\lambda+\i
  0)\psi+\psi',
\end{equation} 
we let (as in \eqref{eq:n2})
\begin{equation}\label{eq:91}
  \widecheck\psi_{0,n}=
  \sum_{ Q_k } \,f_2 (\widecheck H_0){\widecheck Q_k}^*\chi_n\widecheck B_k
  Q_k f_2 (H) R(\lambda+\i
  0)\psi+\chi_n\psi'.
\end{equation}   
\end{subequations}
  By  \eqref{eq:BB^*a}, (\ref{eq:weak}) and  
\eqref{eq:phiweak22}
\begin{equation*}
  \sup_{\epsilon>0}\,\norm [\big]{\widecheck R_0(\lambda+\i
  \epsilon) \parb{\widecheck \psi_{0}-\widecheck\psi_{0,n}}}_{\vB^*}\to 0\text{ for }n\to \infty,
\end{equation*} yielding (\ref{eq:com2b}).

Next we use \cite{IS4} to deduce the asymptotics 
\begin{equation}\label{eq:asres2oneOnb}
       \widecheck  R_0(\lambda+\i 0)\widecheck\psi_{0,n}-2\pi \i \, \breve v^{+}_{\alpha_{\min},\lambda} [g_n]\in \vB^*_0
 \text{ for some }g_n\in \vG_{a_{\min}}=L^2(\mathbf{S}_{a_{\min}}).
    \end{equation} 
Here the second term  is 
  labelled by the `free channel'
    defined uniquely for $a=a_{\min}$ and    denoted by
    $\alpha_{\min}$.
 Note that since the one-body potentials $\widecheck   I_0$ and $\breve
 I_0$ coincide at infinity on any closed conic region  not intersecting collision
 planes, the conditions of  \cite [Lemma 4.10 and Remark  4.11]{IS4} are
 met. This  means
 for the corresponding solutions to the eikonal equation that
 $\lim_{r\to \infty} (\widecheck  K_0-\breve K_0)(r\cdot,\lambda)$ exists locally
 uniformly in $\mathbf{S}_{a_{\min}}\setminus \cup_{a\in \vA_2}  \bX_a$. Here
 the
 solution $\breve K_0=\breve K_{a_{\min}}$
 is  given for the potential $\breve
 I_0$ as in \eqref{eq:quasiM} and \eqref{eq:asres2oneOnb}, while $\widecheck  K_0$ is a similar
 solution for the potential   $\widecheck I_0$. 
 In combination with {\cite [(1.10]{IS4}} applied to $\widecheck   I_0$, we then conclude that  indeed
 the asymptotics \eqref{eq:asres2oneOnb} is fulfilled. 

By combining  \eqref{eq:com2b}, 
  \eqref{eq:asres2oneOnb} and \eqref{eq:eqBnd00}
  we conclude that $g_n\to g\in \vG_{a_{\min}}$
(for some  $g\in \vG_{a_{\min}}$),  which in turn   yields  (by taking $n\to \infty$)
\begin{equation}\label{eq:asres2oneOn}
       \widecheck \Psi_0\phi-2\pi \i \, \breve v^{+}_{\alpha_{\min},\lambda} [g]\in \vB^*_0
 \text{ for some }g\in \vG_{a_{\min}}.
    \end{equation} Consequently     the contribution to $\phi$ from $\Psi_0\phi$  conforms  with
 \eqref{eq:asres29}, as wanted.
  
  \subsection{Completing the proof of Theorem \ref{them:stat-compl-enerMain}}\label{subsec:Conclusion and generalizations}
We have finished the proof of stationary completeness for positive
energies under the
additional condition \eqref{eq:sim2}. 
The general case can  be  treated along the same pattern, to be explained in
this subsection. 

Suppose first that  we drop the condition
\eqref{eq:sim2},  but again consider any  $\lambda\in \R_+\setminus {\vT_{\p}(H)}$. 
We note that  \eqref{eq:sim2}
   was used before for concluding \eqref{eq:ppH} and
   \eqref{eq:thre}. Hence for example we excluded  that  $\lambda\in
   \sigma_{\pupo}(\brH_a)$ (which could occur if $\brh_{a,1}$ has negative eigenvalues
and $H^a$ has eigenvalues above $\lambda$). However we can relax
\eqref{eq:sim2} and avoid this kind of problem  by considering the construction $\brI_{a,R}$ for large
$R$ rather than using only $\brI_{a,1}$ as before:
 Note that  $\inf \sigma(\brh_{a,R})\to 0$ for $R\to
\infty$. Since  $\lambda\notin \vT_\p(H)$, it follows that   $\lambda\notin \vT_{\p}(\brH_{a,R})$ for $R\geq
1$
large, and thus we can use the  Mourre estimate for the modification
$\brH_{a,R}$, rather than just for $\brH_{a,1}$ as done before. In
fact we can repeat the whole analysis from the previous subsections. 
We
conclude that the essential property is  $\lambda\notin \vT_\p(H)$;   the condition
\eqref{eq:sim2} is  not needed for positive energies. 

For negative energies $\lambda\notin \vT_\p(H)$ the same procedure applies. Again we can assume that
$\lambda\notin \vT_{\p}(\brH_{a,R})$ by taking $R$ large enough.  Moreover the contribution to $\phi$
from the terms $\Psi_{a_{\min}}\phi$  and $\Psi_0\phi$ in
\eqref{eq:phasDEC}   conform  with
 \eqref{eq:asres29},
  since in that case  in fact $\breve \Psi_{a_{\min}}\phi=0$  and $\widecheck
  \Psi_0\phi=0$ for     $R$  large in \eqref{eq:Ipott} and 
    \eqref{eq:brevePotential223}, respectively.

The modified  procedure  leads to Theorem \ref{them:stat-compl-enerMain},
as wanted. 
\begin{remark*} 
Of course  
\eqref{eq:sim2}
was used from the beginning of the section when introducing
$\phi=R(\lambda+\i 0)\psi$. If $\lambda\notin \vT(H)$  is an eigenvalue this
$\phi$ is not well-defined. However in this case the corresponding
eigenprojection, say denoted $P_\lambda$, 
maps to $L^2_\infty$ and the expression $(H-P_\lambda -\lambda -\i
0)^{-1}$ (along with its imaginary part) does have an interpretation,
cf. \cite {AHS}, and the scattering theories  for $H$ and
$H-P_\lambda$ coincide. Hence we could deal with $\lambda\in
\sigma_{\pp}(H)\setminus \vT(H)$ upon modifying the Parseval formula
  \eqref{eq:ScatEnergy222331} by using $\delta(H-P_\lambda -\lambda)$
  rather than $\delta(H -\lambda)$ in the formula, cf.  \cite[Remark
  9.3]{Sk1}. With this modified definition of stationary completeness
  at such $\lambda $ our procedure of proof applies (we leave out the
  details). On the other hand $\lambda\notin
\vT(H)$ is an essential condition.

\end{remark*}

\appendix
\section{Green function estimates}\label{sec:Resolvent estimates}
We elaborate on the missing details of our justification of
\eqref{eq:res1}, \eqref{eq:com2}, \eqref{eq:res1233} and
\eqref{eq:com2b}. As in Section \ref{sec:Stationary completeness for
  the 3}  we consider fixed narrowly supported functions $f_3\succ f_2\succ f_1$
such that $f_1=1$ in a neighbourhood of  $\lambda_0$. We recall that in
\eqref{eq:res1} the function  $\psi\in
L^2_\infty$. This  function  $\psi$ is fixed throughout the appendix.

The computation
of $ \brT_a$ in \eqref{eq:res1} yields an expansion into a  sum of terms on a similar form as the ones
considered in the proof of \cite[Lemma 7.3]{Sk1} for  which (in most
cases) \cite[Lemma 2.2]{Sk2} (also appearing in \cite
[Appendix B]{Sk1})
applies.  Hence these terms  are treated by the following list of
$Q$-bounds for which we refer the reader to \cite[Section 7.1]{Sk1}:
\begin{equation}\label{eq:kato10}
  \quad \sup _{\Re z=\lambda_0,\,\Im z >0}\;\norm[\big]{\abs{Q_k f_2
    (H)}{R(z)\psi}}_{\vH}\leq C \norm{\psi}_{\vB},\quad k=1, \dots, 10,
\end{equation} where for any $a\in\vA_1$ 
\begin{align*}
Q_1&=r^{-s},  \quad s\in (1/2,1),\\
Q_2&=r^{-1/2}\sqrt{(\chi^2_+)'}\parbb{ B/\epsilon_0},  \\
Q_3&=Q(a,j)=\xi_j( \hat x)\chi_+(\abs{x})G_{a_j};\q\,\,j\le J(a)\q
(\text {see }\eqref{eq:2boundobtain33},\\
Q_4&=Q_a(b,j)=\xi_j( \hat x)\chi_+(\abs{x})G_{b_j}M_a;\q
b\in\vA_1\setminus\set{a},\, j\le J(b)\q
(\text {cf. }\eqref{eq:2boundobtain33500},\\
Q_5 &=2\parb{\vH^a}^{1/2}\parb{p^a-\tfrac \delta 4\tfrac
  {x^a}r
  B}{f}_3(H) r^{(\rho_1-\delta)/2}T_1^aM_a,\\
&\quad \vH^a={\parb{\mathop{\mathrm{Hess}  }r^a}\parb{x^a/r^\delta}},\\
&\quad \quad T_1^a= \sqrt{(\chi^2_+)'}\parb{ r^{\rho_1/2}B_{\delta}^ar^{\rho_1/2}}{\chi_+}(
B/\epsilon_0),\\
Q_6&=\rho_1^{1/2} r^{-1/2}T^a_2M_a,\\
&\quad T_{2}^a={\zeta_1}\parb{ r^{\rho_1/2}B_{\delta}^ar^{\rho_1/2}}{\eta}(
      B),\\
&\quad \quad \zeta_1(b)=\sqrt{b(\chi^2_+)'(
  b)} , \quad \eta(b)=\sqrt{b}\chi_+(
  b/\epsilon_0),\\
Q_7&=2\parb{\vH^a}^{1/2}\parb{p^a-\tfrac \delta 4\tfrac
  {x^a}r
  B}{f}_3(H)r^{(\rho_1-\delta)/2}T_{1+}^aM_a,\\
&\quad T_{1+}^a= \sqrt{(\chi^2_+)'}\parb{ r^{\rho_1/2}B_{\delta}^ar^{\rho_1/2}}\chi_+\parb{r^{\rho_2-1}r_\delta^a
}{\chi_+}( B/\epsilon_0),\\
Q_8&=\rho_1^{1/2} r^{-1/2}T_{2+}^aM_a,\\
&\quad T_{2+}^a= {\zeta_1}\parb{ r^{\rho_1/2}B_{\delta}^ar^{\rho_1/2}}\chi_+\parb{r^{\rho_2-1}r_\delta^a
      }{\eta}( B),\\
Q_9&={r^{(\rho_2-\rho_1-1)/2}\sqrt{\chi_{1+}}\parb{r^{\rho_2-1}r_\delta^a }
  }T_3^aM_a,\\
&\quad \chi_{1+}(t)=(\chi^2_+)'(t),\\
&\quad \quad  T_3^a =\zeta_2\parb{ r^{\rho_1/2}B_{\delta}^ar^{\rho_1/2} 
}{f}_3(H)\chi_+(
   B/\epsilon_0), \quad\zeta_2(b)=\sqrt{-b \chi^2_+(-b)},\\
Q_{10}&=(1-\rho_2)^{1/2} {r^{-1/2}\sqrt{\chi_{2+}}\parb{r^{\rho_2-1}r_\delta^a }
  }T_4^aM_a\\
&\quad\chi_{2+}(t)=t\chi_{1+}(t),\quad 
T_4^a=\chi_-\parb{ r^{\rho_1/2}B_{\delta}^ar^{\rho_1/2}}\eta(
      B).
\end{align*}

\begin{subequations}
 We  need in addition the following  weak type estimates for
 $\brR_a(z) \breve\psi$ for all $ \breve\psi\in\vB$:
\begin{equation}\label{eq:phiweak}
  \sup _{\Re z= \lambda_0,\,\Im z <0}\;\norm{\abs{{\breve Q_k}f_2
    (\brH_a)}{\brR_a(z) \breve\psi}}_{\vH}\leq
  C\norm{\breve\psi}_{\vB},\quad k=1, \dots, 10,
\end{equation} where 
$\breve Q_k$ are defined as for $ Q_k$ with $f_3(H)$ replaced by
$f_3(\brH_a)$.  These estimates (except for $k=1$) follow by the same commutator scheme 
as   used for \eqref{eq:kato10}, i.e. by \cite[Lemma 2.2]{Sk2}. While
 operators of the form $Q_4$ are not needed for $H$,
 they are for  $\brH_a$, in which case the analogue of
 \eqref{eq:2boundobtain33} stated as  \eqref{eq:kato10} for  operators
 of the form $Q_3$, i.e. \eqref{eq:phiweak} with $k=3$,   is
 established   in (\ref{eq:2boundobtain3350}).

Note that \eqref{eq:kato10} and \eqref{eq:phiweak} have  the interpretation of a uniform bound  on 
 operator norms, in particular  on the operator norm of the adjoint operators
 ${R(\bar z)}f_2(H) Q_k^*$ and ${\brR_a(\bar z)}f_2(\brH_a)\breve
 Q_k^*$\, ($\in\vL(\vH, \vB^*)$).  

Using the notion of order of an operator (recalled in
 Subsection \ref{subsec:A phase-space partition of unity})   we can record that 
  \begin{equation}\label{eq:orderQ}
  Q_k f_2 (H), \,{\breve Q_k}f_2
    (\brH_a)=\vO(\inp{x}^{-s});\quad s=(\delta-\rho_1)/2.
\end{equation}

The
`difficult term' in \eqref{eq:res1} reads,   in terms of  taking  weak limits   in
 $L^2_{-1}$, 
\begin{equation}\label{eq:liB}
  \lim_{\epsilon\to 0_+}  -\i \brR_a(\lambda+\i
 \epsilon )\brT_aR(\lambda+\i \epsilon)\psi  = \lim_{\epsilon\to 0_+}  -\i \brR_a(\lambda+\i
 \epsilon)\brT_aR(\lambda+\i 0)\psi\in \vB^*.
\end{equation} 
The proof of \eqref{eq:liB} and the related property \eqref{eq:com2}
given in   Subsection \ref{subsec:Near a collision
   plane} relies 
on an expansion of $\brT_a$ along the lines of the proof of
\cite[Lemma 7.3]{Sk1}. In the present case this amounts to   expanding
$\brT_a$ into  a  finite sum of terms, each term on one  of 
the respective  ten forms
\begin{equation}\label{eq:Qsquared}
  f_2 (\brH_a){\breve Q_k}^*\breve B_k Q_k  f_2 (H)\text { with 
      }
  \breve B_k\text { bounded}.
\end{equation} 
\end{subequations}

\begin{subequations}
The justification of \eqref{eq:res1233} and the related property
\eqref{eq:com2b} relies on  a similar scheme, in fact we can use
\eqref{eq:kato10} again and replace  \eqref{eq:phiweak} by a similar
(one-body) estimate for  $\widecheck H_0$. Using 
\eqref{eq:parameters} and \eqref{eq:weakerCond} we derive the following substitute  for all $ \widecheck\psi\in\vB$:
\begin{equation}\label{eq:phiweak22}
  \sup _{\Re z=\lambda_0,\,\Im z <0}\;\norm{\abs{{\widecheck Q_k }f_2 (\widecheck
    H_0)}{\widecheck R_0 (z) \widecheck\psi}}_{\vH}\leq C\norm{\widecheck\psi}_{\vB},\quad k=1, \dots, 10,
\end{equation}  where 
$\widecheck  Q_k $ are defined as for $ Q_k $, but with $f_3(H)$ replaced by
$f_3(\widecheck H_0)$. 
We record  that 
\begin{equation}
  \label{eq:orderQ2}
  \widecheck  Q_k f_2 (\widecheck
    H_0)=\vO(\inp{x}^{-s});\quad  s=(\delta-\rho_1)/2.
\end{equation}The
`difficult term' in \eqref{eq:res1233}  reads,   taking
  weak limits in 
 $L^2_{-1}$,  
\begin{equation}\label{eq:tildBB}
  \lim_{\epsilon\to 0_+}  -\i \widecheck R_0(\lambda+\i
 \epsilon )\widecheck T_0 R(\lambda+\i \epsilon)\psi  =\lim_{\epsilon\to 0_+} -\i \widecheck R_0(\lambda+\i
 \epsilon)\widecheck T_0R(\lambda+\i 0)\psi\in \vB^*.
\end{equation}  As above \eqref{eq:tildBB} and \eqref{eq:com2b} require 
expansion of  $\widecheck T_0$ into a sum of
 terms, all being on one of the  forms
\begin{equation}\label{eq:sidste}
  f_2 (\widecheck H_0){\widecheck  Q_k }^* \widecheck B_k\breve  Q_k  f_2 (H)\text { with 
  }
  \widecheck B_k \text { bounded}.
\end{equation}  Again we can mimic the proof of
\cite[Lemma 7.3]{Sk1}. In fact a  term of that proof contains   a
certain  factor
$g(H)-g(\brH_a)$, while  the analogous term in the present case contains a factor 
$g(H)-g(\widecheck H_0)$. The term is  on the form $f_2 (\widecheck H_0){\widecheck  Q_1 }^* \widecheck B_1\breve
Q_1  f_2 (H)$ thanks to the presence of the factors of $\chi_+\parb{
  r^{\rho_1/2}B_\delta^ar^{\rho_1/2}}$ and $\chi_+\parb{
  r^{\rho_2-1}r_\delta^a}$ in the construction of $\widecheck\Psi_0$
and commutation.
\end{subequations}
\section{Non-threshold analysis}\label{sec:Non-threshold analysis}
The recall for the readers convenience a stationary  version of the so-called  minimal velocity bound
obtained directly from velocity bounds in \cite{Is3} and then by a  more simple-minded
stationary argument in \cite{Sk2}. The  notation used below conforms with our
application in (\ref{eq:long}). In particular we assume (\ref{eq:sim2})
and that $f_2$ is a  standard support function obeying  that $f_2=1$
in a neighbourhood of a fixed 
$\lambda>0$.

By a covering
argument it suffices for the third step in (\ref{eq:long}) to show that for any $\psi\in L^2_\infty$ and any
$E>0$ there exists a neighbourhood $U=U_E$ of $E$ such that 
\begin{equation}
  \label{eq:bB2}
  \forall \text{ real }f\in C^\infty_\c(U):\quad \chi_-\parb{
  r^{\rho_2-1}r_\delta^a/2}
f(H^a)\brR_a(\lambda+\i
  0)\psi \in \vB_0^*.
\end{equation} 

The neighbourhood 
 $U$ is determined by  the Mourre estimate for the function $\tfrac 12 (r^a)^2$ (the
one used to define the factor $A_2^a$ in (\ref{eq:propgaObs})):
\begin{equation*}
  1_{U}(H^a)\i\Big [H^a, A^a\Big ]1_{U}(H^a)\geq 2E 1_{U}(H^a);\quad A^a:= \sqrt{r^a} B^a\sqrt{r^a},\,B^a:=2\Re \parb{p^a \cdot \nabla^a r^a}. 
\end{equation*} 

It remains to check (\ref{eq:bB2}). First, we recall (from the factor $A_1$ in
(\ref{eq:propgaObs})) that
\begin{equation*}
  B=\i\big [\brH_a,r\big ]=2\Re \parb{p \cdot \nabla r}.
\end{equation*} The operators $B^a$ and $B$ are  bounded relatively to
$\brH_a$. Now we introduce for any small $\epsilon>0$ the  `propagation observable'
\begin{equation*}
  \Psi=f_2(\brH_a)f(H^a)\zeta_\epsilon \parb{r^{-1/2} A^ar^{-1/2}}f(H^a)f_2(\brH_a),
\end{equation*} where $f\in C^\infty_\c(U)$ is real and $\zeta_\epsilon \in
C^\infty(\R)$ is real and increasing   with  $\zeta'_\epsilon=1$  on
$(-\epsilon,\epsilon)$ and $\sqrt{\zeta'_\epsilon}
\in C_\c^\infty((-2\epsilon,2\epsilon))$.
Since the argument is small on the support of the derivative $\zeta'_\epsilon $
(more precisely bounded by $2\epsilon$), 
we obtain the lower bound
\begin{equation*}\i[\brH_a,\Psi]\geq \tfrac E r f_2(\brH_a)f(H^a)\zeta'_\epsilon \parb{r^{-1/2}A^ar^{-1/2}}f(H^a)f_2(\brH_a)+\vO(\inp{x}^{-2}).
\end {equation*} The error $\vO(\inp{x}^{-2})$ arises from commutation. 

Letting $T_\epsilon=\sqrt{\zeta'_\epsilon} \parb{r^{-1/2}A^ar^{-1/2}}$,
we conclude that for any small  $\epsilon>0$
\begin{equation}
  \label{eq:bB3}
  \forall \text{ real }f\in C^\infty_\c(U):\quad 
T_\epsilon f(H^a)\brR_a(\lambda+\i
  0)\psi \in \vB_0^*.
\end{equation} Now  (\ref{eq:bB2}) follows from (\ref{eq:bB3}) by yet
another commutation argument using the relative boundedness of $B^a$
and the fact that on the support of $\chi_-\parb{
  r^{\rho_2-1}r_\delta^a/2}$
\begin{equation*}
  r^{-1/2}\sqrt{r^a} =\sqrt{r^a}r^{-1/2}=\vO(\inp{x}^{-\rho_2/2}),
\end{equation*} allowing us   after 
inserting  $I=T_\epsilon+(I-T_\epsilon)$ to the left of the factor
$f(H^a)$ in (\ref{eq:bB2}) to conclude  that the second
term $I-T_\epsilon$ contributes by a term
in $\vB_0^*$. Obviously the first term $T_\epsilon$  contributes by a term
in $\vB_0^*$ thanks to  (\ref{eq:bB3}). We have proved
(\ref{eq:bB2}). \qed

\end{document}